\documentclass[11pt,a4paper,reqno]{amsart}
\usepackage{amsmath,amssymb,graphics,epsfig,color,enumerate}
\usepackage{dsfont}  
\usepackage{verbatim}
\definecolor{refkey}{gray}{.75}
\definecolor{labelkey}{gray}{.7}
 \definecolor{darkgreen}{rgb}{0,0.4,0}

\definecolor{light}{gray}{0.9}

\newtheorem{theorem}{Theorem}[section]
\newtheorem{proposition}[theorem]{Proposition}

\newtheorem{corollary}[theorem]{Corollary}

\newtheorem{remark}[theorem]{Remark}

\numberwithin{equation}{section}
\numberwithin{theorem}{section}
\newtheorem{theoremA}{Theorem}

\newtheorem{gigi}{GTUR}

\renewcommand{\epsilon}{\varepsilon}
\renewcommand{\tilde}{\widetilde}
\renewcommand{\hat}{\widehat}

\renewcommand{\div}{\mathop{\rm div}\nolimits}

\definecolor{light}{gray}{.9}

\newcommand{\cA}{\ensuremath{\mathcal A}}

\newcommand{\cE}{\ensuremath{\mathcal E}}
\newcommand{\cF}{\ensuremath{\mathcal F}}

\newcommand{\cJ}{\ensuremath{\mathcal J}}
\newcommand{\cK}{\ensuremath{\mathcal K}}

\newcommand{\cN}{\ensuremath{\mathcal N}}

\newcommand{\cP}{\ensuremath{\mathcal P}}
\newcommand{\cQ}{\ensuremath{\mathcal Q}}
\newcommand{\cR}{\ensuremath{\mathcal R}}

\newcommand{\bbE}{{\ensuremath{\mathbb E}} }

\newcommand{\bbP}{{\ensuremath{\mathbb P}} }

\newcommand{\bbR}{{\ensuremath{\mathbb R}} }

\let\a=\alpha    \let\d=\delta  
 \let\g=\gamma       
            
  \let\s=\sigma \let\t=\tau   
  
   \let\G=\Gamma

\newcommand{\ppp}{\textcolor{black}}
\newcommand{\rrr}{\textcolor{black}}
\newcommand{\red}{\textcolor{black}}
\newcommand{\rosso}{\textcolor{black}}
\newcommand{\blu}{\textcolor{black}}

\newcommand{\ra}{\rangle}
\newcommand{\la}{\langle}

\author[A.C. \ Barato]{A.C.  Barato}
\address{Andre Cardoso Barato
 \hfill\break \indent
 Department of Physics, University of Houston, Houston, Texas 77204, USA
}
 \email{barato@uh.edu}
\author[R.\ Chetrite]{R. Chetrite}
\address{Raphael Chetrite \hfill\break \indent
  CNRS, Laboratoire J.A. Dieudonn\'e, Universit\'e C{\^o}te d'Azur
   \hfill\break \indent
   06108 Nice Cedex 2, France.
 }
 \email{rchetrit@unice.fr}

\author[A.\ Faggionato]{A. Faggionato}
\address{Alessandra Faggionato \hfill\break \indent
  Dipartimento di Matematica, Universit\`a di Roma `La Sapienza'
  \hfill\break \indent
  P.le Aldo Moro 2, 00185 Roma, Italy}
\email{faggiona@mat.uniroma1.it}

\author[D.\ Gabrielli]{D. Gabrielli}
\address{Davide Gabrielli \hfill\break \indent
  DISIM, Universit\`a dell'Aquila
  \hfill\break\indent
  Via Vetoio,  Loc. Coppito, 67100 L'Aquila, Italy}
\email{gabriell@univaq.it}

\thanks{This work   has been  supported  by  Italian PRIN
  20155PAWZB ``Large Scale Random Structures" and by
    the project  ``Investissements d'Avenir'' UCA JEDI of the
French  ANR  n. ANR-15-IDEX-01.
 }

\title[Generalized  thermodynamic uncertainty relations]{A unifying picture of generalized  thermodynamic uncertainty relations}
\begin{document}

\maketitle

\begin{abstract}
The thermodynamic uncertainty relation  is a universal trade-off relation
connecting the precision of a current with the average dissipation  at large times. For
continuous time Markov chains  (also called Markov jump processes) this relation is valid in the  time-homogeneous case, while it fails in the  time-periodic case. The latter   is relevant for the  study of several small
thermodynamic systems.
 We consider here  a time-periodic
Markov chain with continuous time and a  broad class of functionals of stochastic trajectories, which are
general linear combinations of the empirical flow and the empirical density.
Inspired by the analysis  done in our previous work \cite{BaCFG}, we provide general
methods to get local quadratic   bounds for  large deviations, which lead to
universal lower bounds on the ratio of the   diffusion coefficient to the
squared average value in terms of  suitable universal rates, independent of  
the empirical functional.  These bounds  are called ``generalized thermodynamic
uncertainty relations'' (GTUR's), being generalized versions of the thermodynamic
uncertainty relation to the   time-periodic case and to  functionals which  are more general than   currents.
 Previously,  GTUR's in the time-periodic case  have been obtained  in \cite{BaCFG,KSP,PvdB}. Here we recover the   GTUR's in  \cite{BaCFG,KSP} and  produce new ones, leading
to even stronger bounds and also to  new trade-off relations for time-homogeneous systems.
 Moreover, we generalize  to arbitrary protocols the GTUR  obtained in \cite{PvdB} for time-symmetric
protocols. We also generalize to the time-periodic case the GTUR obtained in   \cite{G}   for the so called dynamical activity, and provide a new GTUR  which, in the time-homogeneous case, is  stronger than the one in \cite{G}.
The unifying picture is completed with
a comprehensive comparison between the different GTUR's.   
 \end{abstract}

\section{Introduction}
  The \emph{thermodynamic uncertainty relation} (TUR) recently introduced in \cite{BS1}  is a universal inequality that relates the precision of any current, such as the velocity of a molecular motor or the electron flux
in a quantum dot, with the entropy production that quantifies energy dissipation. More precisely,   the ratio  of the asymptotic diffusion coefficient of any current  to its squared asymptotic value is lower bounded by the inverse average entropy production rate.   This relation constitutes a key result in stochastic thermodynamics  \cite{Se,Sek}, a theoretical framework that extends thermodynamics to small nonequilibrium systems. More generally, the TUR is a consequence
of a parabolic bound on large deviations (LD) proposed in \cite{GHPE,PBS16}. The proof of this bound, which has been obtained in \cite{GHPE}, comes from the  explicit form  derived  in \cite{BFG1,BFG2,MN} of the rate  functional associated with the so called 2.5 level LDs.\smallskip

Several works about the TUR and quadratic bounds on  LD rate functionals  have already been produced  (see for example
\cite{BaCFG,BS0,bisk17,bran18,chiu18,dech18,dech18a,dech18b,TB,G,ging17,GL,hyeo17,HG,KSP,maci18,maes17,NT,NV,nyawo,PBS,trio,PRS,PS100,PNRJ,PLE,PvdB} and references therein). In particular, the  TUR applies to systems driven by a fixed thermodynamic force. Mathematically, these systems can be described as time-homogenous Markov chains, i.e.  with time-independent transition rates,  or  time--homogeneous diffusions as in \cite{GRH,NT,PLE}.
A different way to drive a system out of equilibrium is through an external periodic protocol. Several artificial molecular pumps \cite{ELMN} and colloidal heat engines \cite{MRDR}  constitute experimental examples of such periodically driven systems. A \ppp{continuous--time} Markov chain with time-periodic transitions rates is a standard mathematical framework to describe these systems \cite{BS2}.
\smallskip

 As shown in \cite{BS2},  there is a fundamental difference between systems driven by a fixed thermodynamic force and periodically driven systems concerning the TUR.
The original TUR from \cite{BS1} that involves the entropy production does not apply to periodically driven systems. However, more recently, bounds on current fluctuations that generalize the TUR to periodically driven systems have been obtained in \cite{BaCFG,KSP,PvdB}.
  In  this work  we focus on \emph{generalized thermodynamic uncertainty relations} (shortly, GTUR's). In a very broad sense,
  given a class of empirical functionals, by  GTUR we mean a lower bound on the ratio of the asymptotic diffusion coefficient to the squared asymptotic value of the empirical functional, which holds \emph{uniformly} as the empirical functional varies  in the given class, in the sense that  the lower bounding quantity does not depend on the specific empirical functional and depends only on the Markov process itself and the  class of functionals under consideration.

  \smallskip

A summary of the GTUR's developed so far  (cf. \cite{BaCFG,KSP,PvdB}) is as follows. A first GTUR for
periodically driven systems has been provided in \cite{PvdB}.  This result is restricted to protocols
that are time--symmetric under time reversal and to the  class of empirical functionals fulfilling an antisymmetry relation. The resulting lower bound is in terms of the averaged
entropy production rate, although in a form different from the standard TUR.
 A second contribution has come from  our  previous work  \cite{BaCFG}. There    we have presented a  very  general method to get local quadratic upper bounds on the  LD rate function of currents, and therefore lower bounds on the  ratio of the asymptotic  diffusion coefficient to  the squared asymptotic value. As an application, we have obtained several specific classes of lower bounds (cf. \cite[Eq. (55),(56),(61),(72),(73),(74)]{BaCFG}), which hold for generic currents, also with time-dependent increments (the increment is  the variation of the current  due to a transition).   When restricting to time--independent increments several  lower bounds provided in  \cite{BaCFG} become uniform w.r.t. the possible  increments and therefore are GTUR's, in the sense specified above (cf. e.g. \cite[Eq.~(26),(27)]{BaCFG}).
\ppp{Another} GTUR has been derived in \cite{KSP}  for a class
of empirical functionals given by a current and a generic term that is linear in the
fractions of time spent in a state, the so called empirical density (or measure).

 \smallskip

Part of our main results are an extension of the analysis performed in \cite{BaCFG}. We consider a quite broad class of empirical
functionals. This class includes  currents, which are the standard observables that appear
in the TUR, an observable known as activity that has symmetric increments \cite{G} (in contrast to  currents
that have antisymmetric increments) and the empirical density. In fact, our GTUR's are generalizations of the TUR in two senses: we
consider time-periodic Markov chains and empirical functionals more general
than currents. For instance one of our GTUR's  is a generalization to the time-periodic case
of the bound found in \cite{G} related to the dynamical activity. \ppp{We remark that, even for
  currents and time--homogeneous processes, some of our GTUR's are different and tighter than the usual TUR (similarly, one of our GTUR's is tighter than the bound found in  \cite{G} related to the dynamical activity)}. Finally,
these GTUR's should not be confused with the
generalizations of the TUR to finite time in time-homogeneous, time-inhomogeneous  or time-periodic systems obtained in \cite{dech18,dech18a,dech18b,HG,PRS}.

We  provide   general methods  to produce local quadratic upper bounds on the LD rate function  of the empirical functionals  (cf. Theorems \ref{teo1_metodo},\ref{teo2_metodo} and \ref{teo3_metodo}).  These  methods  rely on the LD principles obtained in \cite{BCFG} and  work whenever one can exhibit a suitable mathematical object, that we call here  \emph{legal input}.  By choosing suitable   legal inputs  we get the  different GTUR's  listed in Section \ref{listone} as \eqref{zac}, \eqref{tanos1},...,\ppp{\eqref{tanos100}}.
In this way we recover the results of \cite{BaCFG,KSP} but also go further, exhibiting new GTUR's which are sometimes  even stronger of the existing ones (for example,  \ppp{\eqref{tanos_teo2}}  provides always a stronger lower bound than the GTUR in \cite{KSP}).

 The GTUR in \cite{PvdB}  is of a different nature. Our unifying picture is completed
with a generalization of this GTUR to the case of general protocols that can be
time-asymmetric \ppp{(cf. \eqref{gtur_exp} in Section \ref{listone})}. This GTUR applies to a class of functionals that fulfills an antisymmetry relation.
Interestingly, the average  entropy production rate that appears in the bound for the case of symmetric
protocols is substituted by an average  naive entropy production rate introduced in \cite{BCFG}. This rate
equals the rate of entropy production plus a rate that becomes zero if the protocol is symmetric.

 All our  results apply as well to time--homogeneous Markov chains with continuous time, since they are a special case of periodically driven systems.
In particular, our GTUR's include the original TUR from \cite{BS1} and imply a generalization of the bound on the fluctuations of activity derived in \cite{G}.

%
%

\bigskip

{\bf Outline of the paper}. In Section \ref{sec_MC} we fix the notation, describe the model and the empirical functionals we will focus on. In Section \ref{listone} we present our main   GTUR's, denoted  by \eqref{zac}, \eqref{tanos1},..., \eqref{gtur_exp}. In Section \ref{sec_esempi} we discuss in detail two examples.  In Sections \ref{sec_TUR_Y} and \ref{sec_TUR_anti} we provide  general methods (cf. Theorems \ref{teo1_metodo}, \ref{teo2_metodo} and \ref{teo3_metodo} there) to get  local quadratic upper bounds on the LD  rate function and derive  all the  GTUR's listed in Section \ref{listone}, apart from \eqref{gtur_exp},  as well as some other lower bounds on the ratio between speed and precision   (cf. Corollaries \ref{ristretto} and \ref{ristretto_bis}).   In Section \ref{sec_proes} we extend the results of \cite{PvdB} to generic protocols (cf. Theorem \ref{teo_proes_extended}),   and derive  \eqref{gtur_exp}.    Finally, we collect some  general remarks and proofs in the Appendixes.

%
%
%
%
%
%
%
%
%
%

\section{Notation and general framework}\label{sec_MC}

\subsection{Models and notation}
We consider a  continuous--time Markov chain $X(t) $ with finite state space $V$ and time--periodic jump rates $w_{ij}(t)$ with period $\tau$:
\[ \bbP(X(t+dt)=j \,|\, \xi(t) =i )=w_{ij}(t) dt\,, \qquad  w_{ij}(t+\tau)= w_{ij}(t)  \qquad \forall i,j\in V\,, \; \forall t\geq 0\,.\]
\rrr{The transition graph associated with the Markov chain $X(t)$  is denoted  $(V,E)$, with vertex set $V$ and set of oriented edges $E$}.
Our main technical  assumptions are the following:
\begin{itemize}
\item[(i)] the graph $(V,E)$ is strongly connected;
\item[(ii)] for each $(i,j)\in E$ it holds $w_{ij}(t)>0$ for all $t$, while for each $(i,j)\not \in E$ it holds \rosso{$w_{ij}(t)=0$} for all $t$.
\end{itemize}
We recall that Item (i) is equivalent to the  fact that, given \rrr{arbitrary} states $i,j \in V$, there exists a path from $i$ to $j$ respecting the edge orientation.

 \smallskip

Denoting by   $P_i(t)$   the probability that  the Markov chain is at state $i$ at time $t$, the time evolution of $P_i(t)$ is given by the equation
\begin{equation}\label{evoluzione}
\frac{d}{dt} P_i(t) =\sum _{j:j\not =i} \bigl[ P_j (t) w_{ji} (t)- P_i(t) w_{ij} (t) \bigr]\,.
\end{equation}
\rrr{The asymptotic properties related to this equation are as follows} (cf. e.g.  \cite{BCFG} for details). In the long time limit, $P_i(t)$ tends to an invariant time--periodic distribution $\pi_i (t)= \pi_i(t+\tau)$. The distribution  \rrr{$\pi(t)$}  can be characterized as the unique invariant  distribution of the discrete--time Markov chain \rrr{$\bigl(X(t+n \tau) \bigr )_{n \geq 0}$}. Other important quantities are
the asymptotic elementary flow $\cQ_{ij}(t)$ and current  $\cJ_{ij}(t)$ along the edge $(i,j)$,  \rrr{which are}  given by
\begin{equation}\label{sale}
\begin{cases}
\cQ_{ij}(t):= \pi_i(t) w_{ij}(t)\,,\\
\cJ_{ij}(t) :=   \pi_i(t) w_{ij}(t)- \pi_j(t) w_{ji}(t)= \cQ_{ij}(t)-\cQ_{ji}(t)\,.
\end{cases}
\end{equation}
Note that  $\pi_i (t)>0$ for all $t >0$ and $i\in V$. \rosso{Moreover} $\cQ_{ij} (t) >0$ for all $t>0$ if $(i,j)\in E$, while  $\cQ_{ij} (t) =0$ for all $t>0$ if $(i,j)\not \in E$.

From equation \eqref{evoluzione} we get the continuity equation
\begin{equation}\label{calma}
\partial_t  \pi_i(t) + \sum _{j: j\not =i } \cQ_{ij}(t)-\sum _{j: j\not=i} \cQ_{ji}(t)=0 \qquad \forall i \in V\,,
\end{equation}
which is equivalent to
\begin{equation}\label{roteante}
\partial_t  \pi_i(t) + \sum _{j: j\neq i} \cJ_{ij}(t)=0\qquad \forall i \in V\,.
\end{equation}
The continuity equation \eqref{calma} can \rrr{be rewritten with a div operator in the form}
\begin{equation}\label{cont_eq}
\partial _t \pi (t) + \div \cQ (t) =0\,,
\end{equation}
where $\pi(t)$ and  $\div \cQ(t) $ are vectors with components $\pi_i(t)$ and  $\div _i \cQ(t):=  \sum _j \cQ_{ij}(t) - \sum_j \cQ_{ji}(t)$.


\rrr{Time independent transition rates $w_{ij}(t)=w_{ij}$ correspond to a particular case
of our theory. In this case, we have a steady state characterized by the asymptotic distribution $\pi$, which fulfills the continuity equation
  \begin{equation}\label{calma2000}
 \sum _{j: j\not =i } \cQ_{ij} -\sum _{j: j\not=i} \cQ_{ji}=0 \qquad \forall i \in V\,,
\end{equation}
where $\cQ_{ij}=\pi_i w_{ij}$.}

   \smallskip

 \blu{Finally,   when the graph $(V,E)$ contains an edge  $(i,j)$ if and only  if
 it contains the edge $(j,i)$, we denote by $\s$ the average entropy production rate. In particular, we have
\begin{equation}\label{entprod1}
\s= \frac{1}{2}\sum_{(i,j)\in E } \frac{1}{\t}  \int_0^\t \cJ_{ij} (t) \ln \frac{\cQ_{ij}(t)}{\cQ_{ji}(t)}dt \,.
\end{equation}
When the transition rates are time-independent, the above identity simply reads
\begin{equation}\label{entprod0}
\s=\frac{1}{2}\sum_{(i,j)\in E }\cJ_{ij}   \ln \frac{\cQ_{ij}}{\cQ_{ji}} \,.
\end{equation}
}
   \smallskip

 Let us introduce the notations for time average and scalar products used in this
paper. In what follows,  when referring to a time--periodic function $f(t)$, we understand  that its period equals $\tau$. Moreover, we denote by $\overline f$ the average of $f$  over a period, \rrr{i.e.}
\begin{equation*}
\overline{f}:= \frac{1}{\tau} \int_0 ^\tau f(t) dt \,.
\end{equation*}
\rrr{The scalar product of two vectors $a(t)$ and $b(t)$ with entries parameterized by $i\in V$ is
given by}
\begin{equation*}
\la a(t), b(t)\ra:= \sum_{i\in V} a_i(t) b_i(t)\,;
\end{equation*}
while, if $a(t)$ and $b(t)$ are matrixes  with entries parameterized by $(i,j) \in V\times V$, \rrr{their scalar product is given by}
\begin{equation*}
\la a(t), b(t)\ra:= \sum_{(i,j) \in V\times V} a_{ij}(t) b_{ij}(t)\,.
\end{equation*}

Finally, in what follows  Markov chains will  always be  considered as time--continuous (i.e. as   Markov jump processes), also when not explicitly stated.

\subsection{Empirical functionals}\label{principiante}

We describe now the class of  empirical functionals on which we will focus and state the associated large deviation principle.
Given  a time-periodic matrix $\a(t)= \bigl( \a_{ij}(t)\,:\, (i,j)\in V\times V\bigr)$ and a   time--periodic vector  $\g(t) = \bigl(  \g_i(t)\,:\, i \in V \bigr) $
we consider the empirical functional $Y_{ \a,\g}^{(n)}$ \rrr{defined as}
\begin{equation}\label{gattino}
Y_{ \a,\g}^{(n)}:= \frac{1}{n \tau} \sum _{\substack{ t\in (0, n\tau]:\\ X(t-)\not = X(t+) }  } \a_{X(t-), X(t+) }(t)+\frac{1}{n \tau} \int_0 ^{n \tau} \g_{X(t) }(t) dt \,.
\end{equation}
\rrr{For example, if all components of $\g(t)$  are zero and the increments $\a_{ij}(t)$ are antisymmetric, i.e.   $\a_{ij}(t)=-\a_{ji}(t)$,
 then $Y_{\a,\g}^{(n)}$ is a
current, which is a key observable in stochastic thermodynamics. If the components of $\a(t)$  are zero, the component $\g_i(t)=1$ and the other components of $\g(t)$ are zero,
then $Y_{\a,\g}^{(n)}$  is the fraction of time spent in state $i$.}

Note that, as $n \to \infty$, $Y_{ \a, \g }^{(n)}$ has the following asymptotics (cf. \cite[Proposition 7.3]{BCFG}):
\begin{equation}\label{LLN}
Y_{ \a,\g }^{(n)} \to y_{\a,\g}\rrr{:=} \overline{ \la \a , \cQ  \ra}+\overline{\la \g, \pi \ra}
\,.
\end{equation}
In particular, if   $\alpha_{ij}= \ln (w_{ij}/w_{ji})$  and $\gamma=0 $,
then $y_{\alpha,\gamma}$ equals the average entropy production rate $\s$ in \eqref{entprod1}.

 As a byproduct of  the large deviation (LD)  principle given by  \rrr{\cite[Theorem 2]{BCFG}} and the  contraction principle (cf. e.g.  \rrr{\cite{DZ,DeS,dH,T}}),   $Y_{\a,\g}
^{(n)}$  satisfies  an LD principle as $n\to \infty$ with \ppp{speed $n\t$}.  Calling $I_{\a,\g}$ its rate functional,  roughly it holds
\begin{equation}
\bbP( Y_{\a,\g}^{(n)} \approx y ) \asymp e^{ -\ppp{n \t} I _{\a,\g} (y) }\,,  \qquad y \in \bbR\,, \; n\gg 1\,.
\end{equation}
We point out  that $I _{\a,\g} (y) \geq 0$ and $I _{\a,\g} (y)=0$ if and only if $y=y_{\a, \g}$. This corresponds to the fact that $y_{\a, \g}$ is the typical value and different values of the functional are \rrr{exponentially unlikely}.

\smallskip

  To describe the variational characterization of  the LD  rate functional  $I_{\a,\g}$\rrr{,} we  introduce the function  $\Phi(q,p)$ defined for $q,p\geq 0$ as
\begin{equation}\label{def_phi}
\Phi(q,p):= q\ln (q/p) -q+p\,,
\end{equation}
with the convention that $\Phi(0,p):=p$ and $\Phi(q,0)=+\infty$ for $q>0$.
Then, it holds
 \begin{equation}\label{variazionale}
I_{\a, \g}(y)= \inf \{ \rosso{I(Q,\rho)}\,:\, (Q,\rho)  \in \cF_{\a, \g,y}  \}\,,
\end{equation}
where
\begin{equation}\label{def_F}
 \rosso{I(Q,\rho)}:=\sum_{(i,j)\in E}  \overline{\Phi \bigl( Q_{ij}(t), \rho_i(t)w_{ij}(t) \bigr) }
\end{equation}
and $\cF_{\a,\g,y}$ denotes the family of pairs  $(Q,\rho)=\left( Q(t), \rho (t)\right)_{t\geq 0}$ such that
\begin{itemize}
\item[(i)]
 $Q(t)$ is a time--periodic flow, i.e. $Q(t)=Q(t+\t)$ and $Q(t)$ is a non--negative function on $V\times V$ which is zero outside $E$  for each time  $t$;
 \item[(ii)]  $\rho(t)$   is a time--periodic probability measure on $V$;
 \item[(iii)] \rrr{the} continuity equation $\partial _t \rho (t) + \div Q(t)=0$ is satisfied, \rrr{where $\div_i Q(t):= \sum _j Q_{ij}(t)- \sum_j Q_{ji}(t)$;}
  \item[(iv)] \rosso{$y =\overline{ \la \a, Q \ra}+\overline{\la \g, \rho \ra}$}.
\end{itemize}

\medskip

We point out
 that  one recovers from \eqref{variazionale} that
$I_{\a, \g}(y_{\a,\g})=0$ \rosso{since, denoting by $\cQ= (\cQ(t))_{t\geq 0}$ and $\pi = (\pi(t))_{t\geq 0} $ the asymptotic flow and density,  respectively, it holds  $I(\cQ, \pi)=0$  in addition to  \eqref{LLN}}.
\medskip

Formula  \eqref{def_F} corresponds to the joint  LD rate functional  of the empirical flow and measure.  To recall their definition, given $t\geq 0$ we denote by $[t]$ the only number in $[0,\t)$ such that $t-[t]$ is a multiple of $\t$. Then
the empirical flow $Q^{(n)}$ is defined as the  measure on $E\times [0,\t)$ given by
\[
Q^{(n)} (i,j, A) := \ppp{\frac{1}{n } } \sharp\left\{  t\in (0, n\tau]:  X(t-)=i\,,\; X(t+)=j\,, \;[t]\in A \right\}\,,
\]
where $\sharp$ denotes the cardinality of the set. On the other hand, the empirical measure $\rho^{(n)}$  is defined as the measure  on $V\times [0,\t)$ such that
\[
\rho^{(n)} (i, A) := \ppp{\frac{1}{n }} \int_0^{n\t}   \mathds{1}\left( X(t)=i\,, \;[t]\in A \right)dt\,,
\]
where $\mathds{1}(\cdot)$ denotes the characteristic function (i.e. the function  equals  $1$ if the event under consideration takes place, otherwise it equals   zero).
Note that, given a time-periodic flow \rrr{$Q=( Q(t) )_{ t\geq 0}$}, we can think of $Q$ as  the  measure on $E\times [0,\t)$ with weights $(i,j,dt)\mapsto   Q_{ij} (t)   dt $. Given a time--periodic probability measure \rrr{$\rho= (\rho(t) )_{t\geq 0}$} on $V$ we can think of $\rho$ as the measure on $V\times [0,\t)$ with weights  $(i,dt)\mapsto \rho_i(t) dt$. In \cite[Theorem 2]{BCFG} it is proved that
the pair $\left( Q^{(n)} , \pi^{(n)}\right)$ satisfies a LD principle with speed \ppp{$n\t$} and rate functional $I(Q,\rho)$ given by \eqref{def_F}
 if  $(Q,\rho)=\left( Q(t), \rho (t)\right)_{t\geq 0}$ satisfies the above conditions (i), (ii), (iii). If these conditions are not fulfilled, then
 $I(Q,\rho)$ equals \rrr{infinity}.  Since
 \begin{equation}\label{rappresento}
 Y^{(n)}_{\a,\g}=\ppp{\frac{1}{\t}} \sum _{i,j} \int_{\rrr{[0,\t)}}  \a_{ij}(t) Q^{(n)} (i,j, dt)+\ppp{\frac{1}{\t}} \sum_i \int _{\rrr{[0,\t)}}  \g_{i}(t) \rho^{(n)} (i ,dt)\,,
 \end{equation}
 \eqref{variazionale} follows from the contraction principle and the above LD principle for $\left( Q^{(n)} , \pi^{(n)}\right)$.


\medskip

The asymptotic diffusion coefficient   $D_{\a,\g}$ associated with $Y^{(n)}_{\alpha,\gamma}$ is defined as
\begin{equation}\label{pizza78}
2 D_{\a,\g}  := \lim _{n \to \infty} n \tau \text{Var}\left(Y^{(n)} _{ \a,\g}\right)\,.
\end{equation}
This quantity
can be obtained from the rate functional $I_{\alpha,\gamma}$
by  the identity
\begin{equation}\label{siluro}
2 D_{\a,\g}= \frac{1}{I''_{\a,\g}(y_{\a,\g})}\,,
\end{equation}
where $I''_{\a,\g}$ denotes the second derivative of $I_{\a,\g}$.
We point out that  in the mathematical literature  the  asymptotic diffusion coefficient is defined without the factor $2$ in the l.h.s. of   \eqref{pizza78}.

\medskip

Formula \eqref{siluro} can be applied when the rate function  $I_{\a,\g}$ is twice differentiable around its minimum point $y_{\a,\g}$. If the set $\mathcal F_{\a,\g,y}$ defined after \eqref{def_F} is non--empty for any real value $y$, then the differentiability could be proved using the smoothness of the function \eqref{def_phi} on points with strictly positive coordinates and the linearity of the constraint $(iv)$ in the definition of $\mathcal F_{\a,\g,y}$. There are however exceptional cases when this does not happen. As an example consider  the case when $\g\equiv 0$ and $\alpha_{ij}=f_j-f_i$ for fixed time--independent constants  $(f_i)_{i\in V}$. Given $(Q,\rho)\in \cF_{\a,\g,y}$ we deduce  that $\div\overline Q=0$  by integrating  the continuity equation $\partial _t \rho (t) + \div Q(t)=0$ on  a period   and using that
$\rho(t)$ is periodic.  Due to the gradient representation $\a_{ij}=f_i-f_j$ and  since
$\div\overline Q=0$,
by a discrete integration by parts  we obtain that  $\langle \a, \overline{Q}\rangle=0$. We get  therefore that $\mathcal F_{\a,\g,y}=\emptyset$ for any $y\neq 0$ (indeed, property (iv) in the definition of $\cF_{\a,\g,y}$ cannot be fulfilled for $y\not =0$). As a consequence $I_{\a,\g}(y)=+\infty$ for $y\not =0$ and $I_{\a,\g}(0)=0$, hence $I_{\a,\g}$ is not differentiable. In this case formula   \eqref{siluro} cannot be applied. See also Remark \ref{zanna_bianca} for another exceptional class.

\section{Main results}\label{listone}


In this section we present our  main GTURs, \blu{given by inequalities \eqref{zac}, \eqref{tanos1},...,\eqref{gtur_exp} below.  Apart from \eqref{gtur_exp},} which is a generalization  of the result derived in \cite{PvdB},
their derivation is obtained by extending the methods and ideas from \cite{BaCFG}. In particular, in  Sections \ref{sec_TUR_Y} and
\ref{sec_TUR_anti} we provide  general methods  to get GTUR's for $\a$ generic and $\a$ antisymmetric, respectively (cf. Theorem \ref{teo1_metodo}, \ref{teo2_metodo}
 and \ref{teo3_metodo}). The results  \eqref{zac},  \eqref{zac_anti} and \eqref{liberazione} presented below are a special case of a class of GTUR's obtained in  Corollaries \ref{ristretto},  \ref{ristretto_bis} and \ref{ristretto_tris}, respectively. Therefore, we refer to Sections \ref{sec_TUR_Y} and \ref{sec_TUR_anti} for more results
and proofs. The extension of the GTUR from \cite{PvdB}  to asymmetric protocols is provided in Section \ref{sec_proes} \blu{(cf. Theorem \ref{teo_proes_extended})}.

\smallskip

\blu{From now on, without further mention, we restrict to the case  that the asymptotic value $y_{\a,\g}$ of the empirical functional $Y^{(n)}_{\a,\g}$ is non zero (see Remark \ref{fiorellino} for the case $y_{\a,\g}=0$).}

\medskip

\subsection{GTUR with generic increments}\label{listone1}
 From Corollary \ref{ristretto}, which contains a more general result,
 we obtain:
 \begin{gigi}
  If the increments $\a$ are time--independent (i.e. $\a_{i,j}(t)\equiv \a_{i,j}$)
  and $\g\equiv 0$, then
 \begin{equation}\label{zac}
\frac{ D_{\a,0}}{y_{\a, 0}^2}\geq\frac{1}{ \hat{\s}}\,, \tag{GTUR 1}
\end{equation}
where
\begin{equation}
\hat{\s}:=  2 \sum _{(i,j)\in E} (\overline{\cQ}_{ij})^2  \overline{\frac{1}{\cQ_{ij}}}\,.
\end{equation}
\end{gigi}
For the case of time-homogeneous Markov chains, $\hat{\s}= 2 \sum _{(i,j)  \in E}\cQ_{ij}$, and this GTUR becomes \blu{\cite[Eq.~(19)]{G}}.
Hence, \eqref{zac} is a generalization of this inequality to time-periodic Markov chains. The quantity $\sum _{(i,j)  \in E}\cQ_{ij}$, which
is the rate of average number of transitions, is known as \blu{\emph{dynamical activity}}. \blu{For time-periodic Markov chains, due to Jensen's inequality, we have   the bound}
$\hat{\s}\ge 2 \sum _{(i,j)  \in E}\overline{\cQ}_{ij}$, \blu{i.e.} $\hat{\s}/2$ is larger than the dynamical activity.

From Corollary  \blu{\ref{teo1_TUR} we} obtain:
\begin{gigi}
For generic increments $\a$, it holds
\begin{equation}\label{tanos1}
\frac{
D_{\a,\g}}{y_{\a, \g}^2}\geq \frac{1}{C(p)}\,,
\tag{GTUR 2}
\end{equation}
 where $p=(p_i)_{i\in V}$ is any  probability on $V$ with $\la  \overline \g, p\ra=0$ and  \begin{equation}\label{cipi}
C(p):= 2 \sum_{(i,j)\in E}  \rosso{p_i^2 \overline{\left( \frac{  w_{ij} ^2}{ \cQ_{ij}   }\right)}}
=2 \sum_{(i,j)\in E}  \rosso{p_i^2 \overline{\left( \frac{  w_{ij} }{ \pi_i   }\right)}}  \,.
 \end{equation}
 \end{gigi}
 Note that the above probability $p$ is time-independent.
This novel GTUR is valid for generic linear functionals  of the form (\ref{gattino}), including the
case $\a=0$, which corresponds to functionals that depend only on the empirical density.

\subsection{GTUR with \blu{antisymmetric} increments}\label{listone2}
In this subsection we assume, without further mention, that
\[ (y,z) \in E\Leftrightarrow (z,y)\in E\,.\]
For the particular case of antisymmetric increments $\a_{i,j}(t)= -\a_{j,i}(t)$, we have the following GTUR's.

 \blu{First}, from Corollary \ref{ristretto_bis}, which contains a more general result,
 we obtain:
\begin{gigi}
  If $\a$ is time--independent \blu{and antisymmetric}
  and $\g\equiv 0$, then
 \begin{equation}\label{zac_anti}
\frac{\rosso{D_{\a,0}}}{y_{\a, 0}^2}\geq \frac{1}{\tilde{\s}}\,,
\tag{GTUR 3}
\end{equation}
where
\begin{equation}
\tilde{\s}:= \sum _{(i,j)\in E } (\overline{\cJ}_{ij})^2  \,\overline{\frac{1}{\cQ_{ij}+ \cQ_{ji}}}\,.
\end{equation}
\end{gigi}This GTUR corresponds to \blu{the first bound in} \cite[Eq.~(27)]{BaCFG}.

Second, from Corollary  \ref{teo2_TUR}, we obtain:

\begin{gigi} \blu{For generic antisymmetric increments $\a$, it holds}
\begin{equation}\label{tanos_teo2}
\frac{ D_{\a,\g}}{y_{\a, \g}^2}\geq \frac{1}{C_{\rm a}(p)}\,,\tag{GTUR 4}
\end{equation}
 where $p=(p_i)_{i\in V}$ is any  probability on $V$ with $\la  \overline \g, p\ra=0$ and
 \begin{equation}\label{cipia}
C_{\rm a}(p):=
\sum_{(i,j)\in E} \overline{\left( \frac{ \bigl(p_i w_{ij}  -p_j w_{ji} \bigr)^2}{\cQ_{ij}  + \cQ_{ji} }\right)}\,.
 \end{equation}
 \end{gigi}

Third, from  Corollary  \ref{ristretto_tris}, which contains a more general result, we obtain:

\begin{gigi} If $\a$ is time--independent \blu{and antisymmetric}
 and  $\g\equiv 0$, then
\begin{equation}\label{liberazione}
\frac{ D_{\a,0}}{y_{\a, 0}^2}\geq \frac{1}{{\s}^*}\,,
\tag{GTUR 5}
\end{equation}
 where
   \begin{equation}\label{cocomero}
  \s^*:=\blu{\frac{1}{2} \sum_{(i,j)\in E } } (\overline{\cJ}_{ij})^2 \overline{\left(\frac{1}{\cJ_{ij}} \ln \frac{\cQ_{ij}}{\cQ_{ji}}\right)} \,.
    \end{equation}
    \end{gigi}

This GTUR  corresponds to \blu{the second bound in} \cite[Eq.~(27)]{BaCFG}. Furthermore,  due to the inequality $\s^*\geq \tilde{\s}$,
which has been proved in \cite{BaCFG}, \eqref{liberazione} \blu{can be also derived directly} from \eqref{zac_anti}. The original TUR for time-homogeneous Markov
chains is a particular case of \eqref{liberazione}. For a time-homogeneous Markov chain $\s^*$ becomes the \blu{average entropy production rate} $\s$
in \eqref{entprod0}.

Fourth, from \blu{Corollary} \ref{gnocchi1}, we obtain:
\begin{gigi} \blu{For generic antisymmetric increments $\a$, it holds}
\begin{equation}
\label{tanos100}
\frac{ D_{\a,\g}}{y_{\a, \g}^2}\geq \frac{1}{C^*_{\rm a}(p)}\,,
\tag{GTUR 6}
\end{equation}
 where $p=(p_i)_{i\in V}$ is any  probability on $V$ with $\la  \overline \g, p\ra=0$ and
 \begin{equation}\label{cipia_star}
C^*_{\rm a}(p):= \frac{1}{2}\sum_{(i,j)\in E} \overline{\left( \frac{ \bigl(p_i w_{ij}  -p_j w_{ji} \bigr)^2}{\cJ_{ij}  }\right)\ln \frac{ \cQ_{ij} }{\cQ_{ji}}}\,.
 \end{equation}
 \end{gigi}
This GTUR, for the particular case $\overline \g=0$ (which is equivalent to the fact that $\g_i(t) = \frac{d}{dt} g_i(t)$ for  periodic functions $g_i$),  has been obtained in \cite{KSP} with a different derivation (cf. \blu{\cite[Eq.~(14),(15),(16)]{KSP}}).

The inequality \eqref{tanos100}  can be derived by the  general method presented in Theorem \ref{teo3_metodo}  as well as directly  from \eqref{tanos_teo2} by the bound \eqref{elementare} presented in  Section \ref{bingo}.

\subsection{GTUR with naive entropy production}\label{listone3}
Our last GTUR follows from Theorem  \ref{teo_proes_extended}, \blu{which contains more general results:}
\begin{gigi}
If $\a_{i,j}(t)= -\a_{j,i}(\t-t)$ and $\g_i (t) = - \g_i(\t-t)$ for any $i,j$ and $t\in [0,\t]$, then
 \begin{equation}\label{gtur_exp}
\frac{ D_{\a,\g}}{y^2_{ \a, \g}}\geq \frac{\tau}{  e ^{\tau \s_{\rm naive}}-1}
\,,\tag{GTUR 7}
\end{equation}
where
\begin{equation}\label{sigmanaive}
\begin{split}
\sigma_{\rm naive} & := \frac{1}{\tau}  \sum_{(i,j)\in E} \int_0^\tau  \pi_i(s) \left[w_{ij}(\tau -s ) -  w_{ij}(s)   \right] ds \\
& \;+\frac{1}{\tau}\sum _{(i,j) \in E}\int_0 ^\tau \pi_i(s) w_{ij}(s) \ln \frac{w_{ij}(s)}{w_{ji}(\tau-s) } ds\,.\end{split}
\end{equation}
\end{gigi}
The above  result is a generalisation  to arbitrary protocols of \cite[Eq.~(2)]{PvdB} for the empirical functionals $Y^{(n)}_{\a,\g}$.   When the period $\t$ is small,  the inverse rate given by the r.h.s. of \eqref{gtur_exp} is well approximated by $1/\s_{\rm naive}$.  Note that
 for  symmetric protocols $\sigma_{\rm naive}=\sigma$. In the limit $\t\to 0$   \eqref{gtur_exp} can be applied only to currents with time--independent increments (due to the constraints $\a_{i,j}(t)= -\a_{j,i}(\t-t)$ and $\g_i (t) = - \g_i(\t-t)$)
  and  it reduces to the classical thermodynamic uncertainty relation $D_{\a,0}/y^2_{\a,0}\geq 1/\s$.




\begin{table}
\begin{center}
\begin{tabular}{ |c|c|c| }
 \hline
 GTUR & Lower bound & Restrictions on $Y^{(n)}_{\alpha,\gamma}$ \\
 \hline
 GTUR 1 & $1/\hat{\s}$  & $\alpha_{ij}(t)=\alpha_{ij}$  and $\gamma_i(t)=0$ \\
  \hline
 GTUR 2 & $1/C(p)$  &   $\la \overline \g, p\ra=0$ \\
 \hline
 GTUR 3 & $\blu{1/\tilde{\s}}$  & $\alpha_{ij}(t)=\alpha_{ij}$, $\alpha_{ij}=-\alpha_{ji}$ and $\gamma_i(t)=0$ \\
 \hline
 GTUR 4 & $1/C_a(p)$  & $\alpha_{ij}(t)=-\alpha_{ji}(t)$ and  $\la \overline \g, p\ra=0$ \\
 \hline
 GTUR 5 & $1/\s^*$  & $\alpha_{ij}(t)=\alpha_{ij}$, $\alpha_{ij}=-\alpha_{ji}$ and $\gamma_i(t)=0$ \\
 \hline
 GTUR 6 & $1/C^*_a(p)$ & $\alpha_{ij}(t)=-\alpha_{ji}(t)$ and  $\la \overline \g, p\ra=0$ \\
 \hline
 GTUR 7 & $\tau/ \left( e ^{\tau \s_{\rm naive}}-1\right)$ & $\alpha_{ij}(t)=-\alpha_{ji}(\tau-t)$ and $\gamma_i(t)=-\gamma_i(\tau-t)$\\
 \hline
\end{tabular}
\vspace{0.5cm}
\caption{Summary of GTUR's written as   $D_{\alpha,\gamma}/y^2_{\alpha,\gamma}\geq \textrm{lower bound}$. The GTUR's are
valid for the linear functionals $Y^{(n)}_{\alpha,\gamma}$ that fulfill the conditions on the third column.}\label{comix}
\end{center}
\end{table}

\subsection{Optimization and comparisons}\label{confronto}
In this subsection we show three propositions. The first is concerned with the
optimal $p$ in the universal rate $C(p)$ in \eqref{tanos1}. The second is concerned
with the relation between \eqref{tanos1}, \eqref{tanos_teo2} and \eqref{tanos100}. The third
is concerned with the relation between \eqref{zac}, \eqref{zac_anti} and \eqref{liberazione}.

We recall that   \eqref{tanos1} holds for any choice of the increments $\a$, antisymmetric  or not.  \rosso{The following result shows the optimal bound \blu{that}  can be \blu{obtained} in  \eqref{tanos1} by taking the minimum \blu{among $p=(p_i)_{i\in V}$}   of $C(p)$:}
\begin{proposition}\label{girasole} Setting
\begin{equation}\label{aiutino}
A_i :=\Big[2  \sum_{j: (i,j)\in E} \overline{\Big(\frac{w_{ij}}{\pi_i }\Big)}\Big]^{-1} \,,
\end{equation}
the optimal bound in \eqref{tanos1} is the following:
\begin{itemize}
\item[(i)] if $\overline\g=0$, \blu{then}
\begin{equation}\label{pollo1}
\frac{ D_{\a,\g}}{y_{\a, \g}^2}\geq\sum _i  A_i\,;
\end{equation}
\item[(ii)] if $\overline \g\not =0$ and  $\overline \g$ has neither  all entries positive nor all entries negative,
\blu{then}
\begin{equation}\label{pollo2}
\frac{ D_{\a,\g}}{y_{\a, \g}^2}\geq
\frac{
\left(\sum_{i}A_i \right)\left(\sum_{i} A_i \overline{\gamma}_{i} ^2
\right)
-\left(\sum_{i} A_i \overline{\gamma}_{i}  \right)^2
}
{\sum_i A_i \overline{\gamma}_{i} ^2} \,,
\end{equation}
and the r.h.s. of \eqref{pollo2} is a positive number.
\end{itemize}

\end{proposition}
For the proof of the above proposition  see  Appendix  \ref{marvel}. We point out that the  optimization among $p=(p_i)_{i\in V}$ for the other constants  $C_{\rm a}(p)$ and $C_{\rm a}^*(p)$ appearing in \eqref{tanos_teo2} and \eqref{tanos100}, respectively, cannot be solved explicitly in the general case.

\begin{remark}\label{ventoso}
When the increments $\a$ are time-independent  and $\g\equiv 0$, one can apply both \eqref{zac} and the optimal \eqref{tanos1} given by \eqref{pollo1}. If the asymptotic density $\pi(t)$  is time--independent  as
in the time--homogeneous case, or as in the time--periodic random walk on the ring considered in Section \ref{anello},  we can prove that \eqref{pollo1} is stronger than \eqref{zac}. We refer to Appendix \ref{marvel} for the derivation.
\end{remark}

When $\a$ is antisymmetric,  we can apply three  $p$--dependent GTUR's, i.e. \eqref{tanos1}, \eqref{tanos_teo2} and \eqref{tanos100}.    Indeed, \eqref{tanos_teo2} is the  optimal one as follows from the next result:
\begin{proposition}\label{fiordaliso} Assume that  $(i,j)\in E$ if and only if $(j,i)\in E$. Then for each probability measure $p=(p_i)_{i\in V}$ it holds
$C_a^*(p) \geq C_a(p)$ and $C(p) \geq C_a(p)$. In particular, when $\a$ is antisymmetric, \eqref{tanos_teo2}  provides the optimal lower bound of $ D_{\a,\g}/y_{\a, \g}^2$ between \eqref{tanos1}, \eqref{tanos_teo2} and \eqref{tanos100}.
\end{proposition}
For the proof  of the above proposition see Appendix \ref{marvel}. The optimality of \eqref{tanos_teo2} stated in Proposition \ref{fiordaliso} is also a consequence of a special alternative  derivation of  this  bound by an optimization procedure (cf. Remark \ref{mushadara}). 

Similarly to Proposition \ref{fiordaliso} we have the following result for the universal constants in \eqref{zac}, \eqref{zac_anti} and  \eqref{liberazione}:
\begin{proposition}\label{violetta}
It holds $\s^*\geq \tilde{\s}$ and $\hat{\s}\geq \tilde{\s}$. In particular, when  $\a$ is  antisymmetric and time-independent and $\g \equiv 0$,  \eqref{zac_anti} provides  the optimal lower bound between  \eqref{zac}, \eqref{zac_anti} and  \eqref{liberazione}.
\end{proposition}
For the proof of the above proposition see  Appendix  \ref{marvel} (we recall that the bound $\s^*\geq \tilde{\s}$ has been derived in \cite{BaCFG}).

\smallskip

  In Section \ref{anello},  considering the case of a random walk on the discrete ring, we show that
  the optimal bound \eqref{tanos_teo2} of Proposition \ref{fiordaliso} and the optimal bound \eqref{zac_anti} of Proposition \ref{violetta} are non--comparable bounds. Similarly  the bounds \eqref{tanos_teo2} and \eqref{gtur_exp} are non--comparable, as well as the bounds \eqref{zac_anti} and \eqref{gtur_exp}. This is illustrated in Figure \ref{brividi} in Section \ref{anello} and corresponds to the crossings of the plotted curves.

\smallskip

We  collect some of the above comparative  results in Table \ref{sfinimento}.

\begin{table}
\begin{center}
\begin{tabular}{|c|c|}
\hline
SET-UP & HIERARCHY OF GTUR's\tabularnewline
\hline
\hline
$\begin{cases}
\pi_i(t)=\pi_i\\
\alpha_{ij}(t)=\alpha_{ij}\\
\gamma_{i}(t)=0
\end{cases}$ &  Optimal GTUR2  \eqref{pollo1} $ \Rightarrow$ \eqref{zac} \tabularnewline
\hline
$\alpha_{ij}(t)=-\alpha_{ji}(t)$ & $\begin{cases}
\eqref{tanos_teo2}\Rightarrow \eqref{tanos1}\\
\eqref{tanos_teo2} \Rightarrow \eqref{tanos100}
\end{cases}$\tabularnewline
\hline
$\begin{cases}
\alpha_{ij}(t)=\alpha_{ij}\\
\alpha_{ij}=-\alpha_{ji}\\
\gamma_{i}(t)=0
\end{cases}$ & $\begin{cases}
\eqref{zac_anti}\Rightarrow \eqref{zac}\\
\eqref{zac_anti} \Rightarrow \eqref{liberazione}
\end{cases}$\tabularnewline
\hline
\end{tabular}
$$\,$$
\caption{Implications between GTUR's}\label{sfinimento}
\end{center}
\end{table}

%
%

\subsection{Further comments on \eqref{gtur_exp} }
The rate $\sigma_{\rm naive}\ge 0$ is the asymptotic average value per unit time of   the functional of the trajectories introduced in \cite[Section 4]{BCFG} and described as follows: the functional  equals the logarithm of the ratio of the weight of the forward trajectory 
to the weight of the backward trajectory, without reversal of the protocol. This situation is different 
from the average entropy production rate  $\sigma\ge 0$, for which the reversed trajectory with reversed protocol is 
considered. Furthermore, the quantity $\sigma_{\rm naive}$ can be written as $\sigma_{\rm naive}=\s+\s_{\rm asy}$, 
where 
\begin{equation}\label{entasy}
\s_{\rm asy}:=\sum_{(i,j)\in E}  \blu{\overline{ \cQ_{ij}(A_{ij}-1-\ln A_{ij})}}\,,
\end{equation}
and $A_{ij}(t):= w_{ij}(\t-t)/w_{ij}(t)$. This decomposition has a nice physical interpretation, 
the average entropy production rate $\s$ quantifies energy dissipation and $\sigma_{\rm asy}$ is zero if 
the protocol is symmetric. We point out that $-\s\le\sigma_{\rm asy}\le \s_{\rm naive}$ and that $\s_{\rm asy}$ can have arbitrary sign, as demonstrated with an explicit calculation in Section \ref{anello}.

If, in addition to  \eqref{gtur_exp}, it is possible to apply   \eqref{zac_anti}  or  \eqref{tanos_teo2}  (for example for currents with time--independent increments),
then  there  is a priori no fixed  order between the corresponding rates. This fact is demonstrated by  an example in Section \ref{anello}.

Finally, \eqref{gtur_exp} does not work well when the periodically driven Markov chain is
 obtained by a weak perturbation of a time-homogeneous Markov chain  with continuous time. Indeed, suppose that the transition 
rates are given by  $w_{ij}(t)= c_{ij} +\epsilon d_{ij} (t)$,  where  $c_{ij} $ are the transition rates 
of an irreducible  time--homogeneous Markov chain   and $d_{ij}(t)$  are  genuinely time periodic, with period $\tau$. 
 Then, when $\epsilon \to 0$,  the value in the r.h.s. of \eqref{gtur_exp}
converges to $\tau( e^{\tau\sigma }-1)^{-1}$, which is   smaller (and even much smaller for $\tau$ large) than  
 $1/\s$ entering  in the standard thermodynamic uncertainty relation.  On the other hand, the
rates $C(p)$, $C_{\rm a}(p)$, $C_{\rm a}(p^*)$, $\hat{\s}$, $\tilde{\s}$, $\s^*$    behave well under perturbations.


\subsection{Comments on the weights $\g$}
 We observe that   \eqref{tanos1}, \eqref{tanos_teo2} and \eqref{tanos100} are uniform among the weights $\g$ such that $\la \g,p\ra=0$ for some probability measure $p$ on $V$.
We remark that one cannot find a GTUR uniform among   all $\g$'s. Indeed, if  we consider new weights $\g'$ defined as $\g'_i= \g_i +c$  for some fixed constant $c$, we get  $y_{\a, \g'}= y_{\a, \g}+c$, while $D_{\a,\g'}= D_{\a, \g}$. In particular  the ratio of the asymptotic diffusion coefficient to  the squared asymptotic value can be made arbitrarily small by playing with $c$. On the other hand, if we  take  $\g'=c \g$ for some $c\not=0$, we have  $y_{\a, \g'}= c\, y_{\a, \g}$, while $D_{\a,\g}=c^2 D_{\a, \g'}$, thus implying that  $D_{\a,\g}/y_{\a,\g}^2=D_{\a, \g'}/y_{\a, \g'}^2$. As a consequence GTUR's are automatically uniform among proportional $\g$'s.
 In  \eqref{tanos1}, \eqref{tanos_teo2} and \eqref{tanos100}  one  goes further  replacing proportionality by the weaker condition $\la\overline \g, p\ra=0$.

We also observe that, given $\g$, the existence of a probability measure $p$ such that $\la \overline \g, p\ra=0$ is equivalent to the fact that  the entries of $\overline \g$ are not all positive and not all negative.  If for example the entries of $\g$ are all positive, by taking a suitable constant $c$  one can apply   \eqref{tanos1}, \eqref{tanos_teo2} and \eqref{tanos100}  to the weights $\a, \g'$ where
$\g'_i= \g_i +c$ , and then recover information on $Y_{\a, \g}^{(n)}$ by using that $Y_{\a, \g}^{(n)}= Y_{\a, \g'}^{(n)}-c$, $y_{\a, \g}= y_{\a, \g'}-c$, $D_{\a, \g}= D_{\a, \g'}$.


\subsection{Comments on the derivation of  the GTUR's}\label{listone4}

\smallskip

We comment the methods used to derive the above GTUR's, which are summarized in Table \ref{comix}:

\begin{enumerate}
\item  Due to \eqref{variazionale} we have the upper bound\begin{equation}\label{missile}
I_{\a, \g} (y) \leq I(Q,\rho) \text{ for any $(Q, \rho) \in \cF_{\a, \g,y}$}\,,
\end{equation}
where  $I(Q,\rho)$ is the explicit  function  given in  \eqref{def_F} and the set $\cF_{\a, \g,y}$ is defined after \eqref{def_F}.

\item \rosso{We assume to have an $y$--parameterized pair $(Q_y,\rho_y)\in \cF_{\a, \g,y}$ \blu{such that}
  $(Q_y,\rho_y)$  differs from  $(\cQ,\pi)$    by a term proportional to $y-y_{\a, \g}$.}

\item Plugging the above $y$--parameterized pair $(Q_y, \rho_y )$ in the inequality \eqref{missile} and taking a 2nd order Taylor expansion of the explicit function  $y\mapsto I(Q_y,\rho _y) $ around $y_{\a, \g}$, one gets a quadratic local bound of  $I_{\a, \g}$ at $y_{\a,\g}$ (cf. Theorem \ref{teo1_metodo} in Section \ref{sec_TUR_Y}). \blu{A lower bound for the ratio $D_{\a,\g}/y_{\a,\g}^2$ can then be obtained  by \eqref{siluro}}.

\item By exhibiting different  choices of $(Q_y, \rho_y)$  satisfying the above general conditions, we obtain  \eqref{zac} and \eqref{tanos1}.

\item When $\a$ is antisymmetric $Y_{\a, \g}^{(n)} $ can be expressed as \blu{a} linear  function  of the  empirical density and  current, whose LD principle has been derived  in \cite{BCFG} with an explicit LD rate functional \rosso{$I_*$}. The above strategy can be implemented working with currents \blu{($J$)}  instead of \blu{flows ($Q$)}. Hence, we get a general result given by Theorem \ref{teo2_metodo},  which is the analogous of  Theorem \ref{teo1_metodo}. Indeed, by a different approach,  we show that Theorem \ref{teo2_metodo}   can  be  even derived from Theorem \ref{teo1_metodo}.  By exhibiting    different  choices of $(J_y,\rho_y)$  satisfying our general conditions, we get \eqref{zac_anti} and   \eqref{tanos_teo2}.

\item  \blu{\eqref{tanos100} is a consequence of a general result detailed in Theorem \ref{teo3_metodo}. This theorem  can be obtained along the above scheme, with the   exception that one uses an upper bound of  the LD rate function $I_*$ by a suitable function proposed by \cite{GHPE} and  afterwards applies a  2nd order Taylor   expansion to this function. We also show that indeed Theorem \ref{teo3_metodo} can also be obtained as corollary of Theorem \ref{teo2_metodo}.}

\item \eqref{gtur_exp}  is  a special case of a more general result
 given in Theorem \ref{teo_proes_extended} in Section \ref{sec_proes} and its derivation  follows very closely the one
  in \cite{PvdB}.
The trajectory of the Markov chain  on the time interval $[0,n\t]$ can be thought  of as a concatenation of paths on the fundamental periods $[0,\t]$, $[\t, 2\t]$,.., $[(n-1)\t,n\t]$. One obtains a  LD principle for the  frequencies of   these paths. On the other hand, the empirical functional $Y^{(n)}_{\a,\g}$ can be expressed  as a linear functional of the above frequencies and by contraction one gets a new variational characterization for $I_{\a,\g}$. By playing with suitable inputs in the variational characterization, one finally gets the resulting  quadratic local upper bounds on $Y^{(n)}_{\a,\g}$ and, as a byproduct with \eqref{siluro}, \eqref{gtur_exp}.

\end{enumerate}

%
%
%
%

\section{Examples}\label{sec_esempi}


We study here two specific examples, given by   a periodically driven 2-state Markov chain and  a periodically driven  random walk on a ring. The latter is particularly relevant for the comparison of GTUR's.

\subsection{2-state model} We consider a periodically driven 2-state Markov chain, which can be used e.g. to study a quantum dot.
We take $V=\{0,1\}$. Then the \rrr{periodic} stationary distribution $\pi_i(t)$ has the following form (cf. \cite[Sec.~6]{BCFG},  \cite[Prop.~3.13]{FGR}):

 \begin{align*}
& \pi_0(t) =\frac{e^{-C(t)}}{1-e^{-C(\t)}}\left[\int_0^t w_{10}(s) e^{C(s)}\,ds+
e^{-C(\t) }\int_t^{\t}w_{10}(s)e^{C(s)}\,ds\right]\,,\\
&   \pi_1(t)=\frac{e^{-C(t)}}{1-e^{-C(\t)}}
\left[\int_0^t w_{01}(s) e^{C(s)}\,ds+
e^{-C(\t)}\int_t^{\t}w_{01}(s)  e^{C(s)}\,ds
\right]\,,
\end{align*}
where
\[
C(t):=\int_0^t\left[w_{01}(s)+w_{10}(s)\right]\,ds\,.
\]
\rrr{For simplicity we restrict below to $\overline\g=0$.}

When $\a$ is arbitrary, by Proposition \ref{girasole} we get the optimal (among $p=(p_0,p_1)$) \eqref{tanos1}
\begin{equation}\label{energia1}
\frac{ D_{\a,\g}}{y_{\a, \g}^2}\geq  \frac{1}{2} \left[  \left(\overline{  w_{01}/\pi_0}\right)^{-1}+\left(\overline{ w_{10}/\pi_1}\right)^{-1}\right]\,.
\end{equation}

 \medskip

 When  $\a$ is antisymmetric, by Proposition \ref{fiordaliso}  we know  that   \eqref{tanos_teo2}  is the optimal one between
  \eqref{tanos1}, \eqref{tanos_teo2} and \eqref{tanos100}. One can optimize $C_a(p)$ among the probabilities $p=(p_0,p_1)$ \rrr{as follows}.
  \rrr{Defining  $T(t)$ as }
  \[ T(t):= \cQ_{01}(t)+\cQ_{10}(t)= \pi_0(t) w_{01}(t) +\pi_1(t) w_{10}(t)\,,\]
  by straightforward computations we get
\begin{equation}\label{pioggia}
  \min \{ C_a(p)\,:\, p=(p_0,p_1)\}=2
    \frac{
    \overline{w_{01}^2/  T}\cdot \overline{w_{10}^2/ T}- \left( \overline{w_{01}w_{10}/T}\right)^2
  }{
  \overline{ (w_{01}+w_{10})^2/T}
  }\,.
  \end{equation}
Note that, if one introduces on the fundamental period $[0,\t]$ the probability measure
\[ \nu(dt):=
\left[
 \int_0^\t \frac{1}{ T(s)} ds
  \right]^{-1}
 \frac{1}{T(t) } dt\,,
\]
then we can think of $w_{01}(t)$ and $w_{10}(t)$ as random variables on the probability space  $([0,\t], \nu)$ and
 the   optimal constant given by the r.h.s. of \eqref{pioggia} equals
  \begin{equation*}
 \blu{2} \frac{ {\rm Cov}_\nu \bigl(w_{01}; w_{10} \bigr) }{\nu \bigl( (w_{01}+w_{10})^2\bigr)}\,.
\end{equation*}
By \eqref{tanos_teo2} we have, for  $\a$  antisymmetric,
\begin{equation}\label{energia100}
\frac{ D_{\a,\g}}{y_{\a, \g}^2}\geq  \frac{1}{2}
    \frac{
  \overline{ (w_{01}+w_{10})^2/T}
  }
{
    \overline{w_{01}^2/ T}\cdot \overline{w_{10}^2/T}- \left( \overline{w_{01}w_{10}/T}\right)^2
  }\,.
\end{equation}

We recall that the GTUR's presented  in Section \ref{listone}
 are meaningful  under the condition that $y_{\a,\g}\not =0$.  If one restricts to time--independent currents (i.e. $\a_{ij}(t)=\a_{ij}$, $\a_{ij}=-\a_{ji}$,  and $\g\equiv 0$), then this condition fails, since  $\overline{\cJ}_{01}=0$ for all $t\geq 0$ (see Remark \ref{zanna_bianca} for a generalization).



%
%
%
%

\subsection{Random walk on the ring}\label{anello} We consider a random walk on a ring with $N$ sites, where $k_+(t)$ and $k_-(t)$ \rrr{are} the periodic  probability \rrr{rates} to make a unitary jump clockwise and anticlockwise, respectively. In this case $\pi_i(t)= 1/N$ by symmetry.

 Due to Proposition \ref{girasole}, when \emph{$\a$ is arbitrary and $\overline\g=0$}, we have the  optimal \eqref{tanos1}
\begin{equation}\label{prp1}
\frac{ D_{\a,\g}}{y_{\a, \g}^2}\geq\frac{1}{2} \frac{1}{\overline{k}_++\overline{k}_-} =: \frac{1}{r}\,.
\end{equation}
By Proposition \ref{fiordaliso}, when \emph{$\a$ is \rrr{antisymmetric}  and $\overline\g=0$},
\eqref{tanos_teo2} is optimal among  \eqref{tanos1}, \eqref{tanos_teo2} and \eqref{tanos100}. By optimizing \eqref{tanos_teo2} among the probability measures $p=(p_i)$, we get that the minimum is attained at the uniform probability  and therefore we get  the optimal \eqref{tanos_teo2}
\begin{equation}\label{prp2}
\frac{ D_{\a,\g}}{y_{\a, \g}^2}\geq\frac{1}{2} \left[ \,\overline{  \left(\frac{(\rrr{k_--k_+})^2}{k_-+k_+}\right) }\, \right]^{-1}=: \frac{1}{r_{\rm a}} \,\,.
\end{equation}

For this model we have
\begin{equation}
\tilde {\s}= 2 (\overline{k}_- -\overline{k}_+)^2 \overline{ \frac{1}{k_-+k_+} }\,.
\end{equation}
and
\begin{equation}
\begin{split}
\s_{\rm naive}&= \overline{  k_+ \ln \frac{k_+}{k_-(\t-\cdot)} + k_- \ln \frac{k_-}{k_+(\t-\cdot)}} \\
& = \s+ \overline{ k_+ \ln \frac{k_-}{ k_-(\t-\cdot)} + k_- \ln \frac{ k_+}{ k_+(\t-\cdot)}}\,.
\end{split}
\end{equation}
\rrr{Above the function $k_\pm (\t-\cdot)$ is defined as $t\mapsto k_\pm (\t-t)$.}
If we take for example
$k_+\equiv 1$ we get
$\s_{\rm naive}= \s+ \overline{ \ln \frac{k_-}{ k_-(\t-\cdot)} }$. This shows that there is not a fixed order between $\s_{\rm naive}$ and $\s$.
Indeed, given a positive periodic function $f$, the random walk with rates $k_+\equiv 1$ and $k_-= f$  and  the random walk with rates $k_+\equiv 1$ and $k_-= f(\t-\cdot)$  have inverted ordering for $\s$ and $\s_{\rm naive}$. In particular, $\s_{\rm asy}=\s_{\rm naive}-\s$ can be positive and    negative as well.

\smallskip

 Let us now  take the following  time--symmetric protocol,  where $a,b,c,d$ are positive numbers:
\begin{equation}\label{zittino}
k_+(t)=
\begin{cases}
a & \text{ if } t \in [0,\tau/4)\\
b & \text{ if } t \in [\tau/4, 3\tau/4) \\
a & \text{ if } t \in [3\tau/4, \tau)
 \end{cases}\quad
\text{ and }\quad
k_-(t)=
\begin{cases}
c & \text{ if } t \in [0,\tau/4)\\
d & \text{ if } t \in [\tau/4, 3\tau/4) \\
c & \text{ if } t \in [3\tau/4, \tau)
\end{cases}
\,.
\end{equation}
Then we have
\begin{align}
r &= a+b+c+d\,, \label{halloween1}\\
r_{\rm a} &= \frac{(a-c)^2}{a+c}+\frac{(b-d)^2}{b+d} \,, \label{halloween2}\\
\tilde{\s}&=\frac{1}{4}[(a+b)-(c+d) ]^2\left(\frac{1}{a+c}+\frac{1}{b+d}\right)\,,\label{halloween3}\\
\sigma_{\rm naive}&= \sigma=   \frac{a-d}{2} \ln \frac{a}{d}+\frac{b-c}{2}\ln \frac{b}{c}\,.\label{halloween4}
\end{align}
Note that all the above quantities do not depend on the  period $\t$.

Due to Proposition \ref{fiordaliso} $r$  is lower bounded by $r_{\rm a}$. Due to Proposition \ref{violetta} $\tilde{\s}$ lower bounds $\hat{\s}$ and $\s^*$.  We concentrate on the comparison between the constants $ r_{\rm a},\tilde{\s}$ (which are optimal in the sense clarified by Propositions \ref{fiordaliso} and \ref{violetta}) and $\s_{\rm naive}$.
As shown in Figure \ref{brividi},
there is  no  fixed order  either  between  $r_{\rm a}$ and $\s_{\rm naive}$ or between $\tilde{\s}$ and $\s_{\rm naive}$.
Note that, for $\t\to 0$, the universal constant $(e^{\t \s_{\rm naive} }-1)/\t$ converges to $\s_{\rm naive}$. As a consequence there is no optimality either between the GTUR  \eqref{prp2} and   \eqref{gtur_exp}
or between \eqref{zac_anti} and \eqref{gtur_exp}.  Figure \ref{brividi} shows also  that there is no fixed order between $r_a$ and $\tilde{\s}$, i.e.  \eqref{zac_anti} and \eqref{tanos_teo2} are non--comparable bounds.

%
%

\begin{figure}\includegraphics[width=10cm]{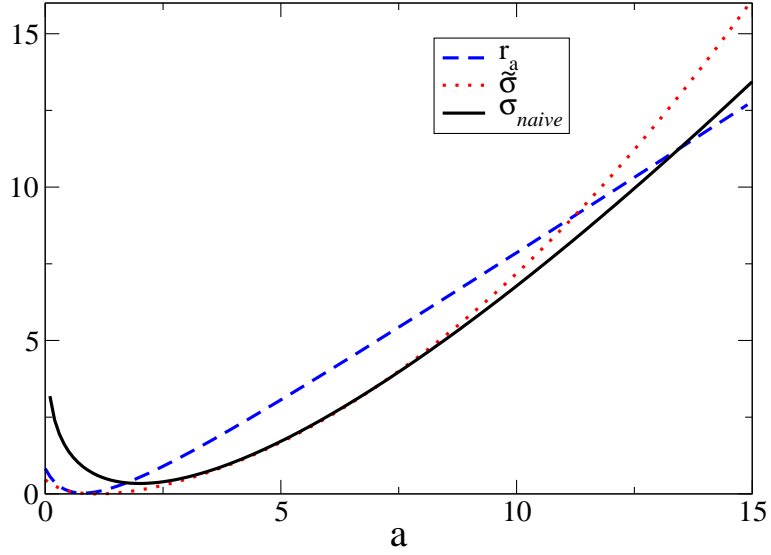}
\caption{The constants $r_{\rm a}$, $\tilde{\s}$ and $\s_{\rm naive}$ as functions of the parameter  $a$ for the random walk on the ring in Section \ref{anello}
with fixed parameters $b=1.7$, $c=0.8$ and $d=2$.}\label{brividi}
\end{figure}

\section{\rosso{Local bounds on $I_{\a,\g}$ and GTUR's for $Y_{\a, \g}^{(n)}$:  case of generic  $\a$}}
 \label{sec_TUR_Y}

Our first aim is  to describe a general method to get  local quadratic upper bounds on $I_{\a, \g}$ around its minimum point $y_{\a,\g}$, thus leading  also to lower bounds on $D_{\a,\g}$ via \eqref{siluro}. This method is an extension of the one  used for the empirical currents in \cite[Section 4.3]{BaCFG}.

\smallskip
\rrr{It is convenient to introduce the concept of  \emph{generalized flow}, which  is defined as a flow without the  restriction of non--negativity. In other words, we will call generalized flow any  function $k: V\times V\to \mathbb R$ which is zero outside $E$.  \rrr{If $k$ is non-negative, then $k$ is a flow}.   The divergence of $k$ is defined as
\begin{equation}\label{pera}
{\rm div}_i k := \sum _j k_{ij}- \sum_j k_{ji}\,.
\end{equation}}

Due to \eqref{variazionale} one has
\begin{equation}\label{sonia1}
I_{\a, \g}(y) \leq \rosso{I(Q, \rho)} \,, \qquad \forall (Q,\rho)\in \cF_{\a, \g,y}\,,
\end{equation}
where the set $ \cF_{\a, \g,y}$ is defined after  in \eqref{def_F}.
Moreover, the function \rrr{$\Phi(q,p)$ defined in \eqref{def_phi}} satisfies the following bound
obtained by a Taylor's expansion around the arbitrary diagonal point $(a,a)$:
\begin{align}
& \Phi(q,p)= \frac{1}{2a}(p-q)^2 + o\left( (q-a)^2+(p-a)^2 \right)\,. \label{local_bb}
\end{align}
Due to \eqref{LLN}, when $y$ is  \rrr{close} to the asymptotic value $y_{\a,\g}$, it is natural  to look for pairs $(Q,\rho)\in \cF_{\a, \g,y}$ which are obtained as perturbation of $(\cQ,\pi)$. To this aim, it is convenient  to use the representation
\begin{equation}\label{denti-a}
\begin{cases}
Q=\cQ+\frac{y-y_{\a,\g} }{y_{\a,\g}} R\,,\\
\rho= \pi+\frac{y-y_{\a,\g} }{y_{\a,\g}} m\,.
\end{cases}
\end{equation}
\rrr{We assume $y_{\alpha,\gamma}\neq 0$ in this equation. For
the case $y_{\alpha,\gamma}= 0$, the equation has to be modified, as explained in Remark \ref{fiorellino} below.}

\medskip

Note that $(Q,\rho)\in \cF_{\a, \g, y}$ if and only if the following properties are satisfied by the pair $(R,m)$:

\smallskip

\begin{itemize}
\item[(P1)] $R= \big(R(t)\big)_{t\geq 0}$ is a  time--periodic generalized flow and therefore  $R(t): V\times V\to \bbR$ is zero outside $E$ for all $t\geq 0$;
\item[(P2)] $m=\bigl( m(t) \bigr) _{t\geq 0}$ is time--periodic and $m(t): V \to \bbR$ satisfies $\sum_i m_i(t)=0$ for all $t\geq 0$;
\item[(P3)] \rosso{$\partial _t m (t)  + \div R(t) =0$}, 
\item[(P4)] \rosso{$y_{\a,\g} =\overline{ \la \a, R \ra}+\overline{\la \g, m \ra}$};
\item[(P5)]   the  functions in the r.h.s. of \eqref{denti-a} take non--negative values.
\end{itemize}

\smallskip

We point out that, given $R,m$ \rosso{satisfying (P1) and (P2)}, since $\cQ_{ij}(t)>0$ for all \rosso{$(i,j)\in E$} and $\pi_i(t)>0$ for all $i\in V$, \rosso{property (P5) is satisfied for $y$} sufficiently close to $y_{\a, \g}$. Since our bounds are local for $y$ close to $y_{\a,\g}$,  we will disregard (P5) in what follows.

\begin{theoremA}\label{teo1_metodo}  For any pair $(R,m)$ fulfilling  the above properties (P1),...,(P4)  the following local quadratic upper  bound holds:
\begin{equation}\label{flash-a}
I_{\a,\g}(y) \leq\frac{1}{2} \frac{(y-y_{\a,\g})^2}{y_{\a,\g}^2} \sum_{(i,j)\in E} \overline{\left(
\rosso{
\frac{ \bigl(R_{ij}- m_i w_{ij} \bigr)^2}{\cQ_{ij}  }
}
\right)}+ o\left(  (y-y_{\a,\g})^2\right)\,.
\end{equation}
In particular, we have the lower bound
\begin{equation}\label{tanos-a}
2 D_{\a, \g}\geq  y_{\a,\g} ^2  \left\{ \sum_{(i,j)\in E}  \overline{
\left(
\rosso{
\frac{ \bigl(R_{ij}- m_iw_{ij} \bigr)^2}{ \cQ_{ij}  }
}
\right)}\right\}^{-1} \,.
\end{equation}
\end{theoremA}
We point out that, since \eqref{denti-a} defines a  bijection $(R,m)\mapsto (Q,\rho)$,  one would get an identity in \eqref{flash-a} and \eqref{tanos-a}  by optimizing among $(R,m)$ in the above theorem.
\begin{proof}
From  \eqref{def_F} and \eqref{local_bb}
by setting  $a=\mathcal{Q}_{ij}(t)$,
we have
\begin{equation}\label{acqua1}
\begin{split}
\rosso{I(Q,\rho)}:&=\sum_{(i,j)\in E}  \overline{\Phi \bigl( Q_{ij}(t), \rho_i(t)w_{ij}(t) \bigr) }\\
& =\sum_{(i,j)\in E} \overline{\left(  \frac{\left ( Q_{ij}- \rho_i  w_{ij} \right)^2}{2 \cQ_{ij}  } +  \cE_{ij}\right ) }
\\
& =\frac{1}{2} \frac{(y-y_{\a,\g})^2}{y_{\a,\g}^2}  \sum_{(i,j)\in E} \overline{\left(  \frac{\left (R_{ij}- m_iw_{ij}   \right)^2 }{ \cQ_{ij} }+  \cE_{ij}\right ) },
\end{split}
\end{equation}
where the error term \rosso{$\cE_{ij}(t)$} is given by
\begin{equation}\label{acqua2}
\begin{split}
\cE_{ij}(t) & =o\left( \,(Q_{ij}(t)- \cQ_{ij}(t) )^2\, \right) +o\left( \,(\rho_i(t)- \pi_{i}(t) )^2\, \right)
= o\left( (y-y_{\a,\g})^2\right)\,.
\end{split}
\end{equation}
Equations \eqref{acqua1} and \eqref{acqua2} imply \eqref{flash-a}. Finally, \eqref{tanos-a} follows from \eqref{flash-a} by means of \eqref{siluro}.
\end{proof}
\begin{remark}\label{fiorellino} When $y_{\a,\g}=0$ the above arguments remain  valid by
making the following changes.
Formula \eqref{denti-a} becomes $Q=\cQ+yR, \rho=\pi+ym$. In (P4) one replaces
 $y_{\a,\g}$ with $1$. In (P5) the $y_{\a,\g}$'s on the numerator are $0$ while the ones on the denominator become $1$. Then  Theorem \ref{teo1_metodo} remains valid by replacing $y_{\a,\g}$ with  $1$ in the denominator of \eqref{flash-a} (the $y_{\a,\g}$ in the numerator is zero) and in \eqref{tanos-a}.
\end{remark}
Theorem \ref{teo1_metodo} provides  a very general method from which several \rosso{local quadratic bounds and GTUR's}
  can be derived \rosso{by inserting different choices of $(R,m)$}. To get  sharp and interesting bounds it is important to select  special perturbations $(R,m)$ fulfilling  the above properties (P1)--(P4).
We discuss  below some special choices, leading to  some corollaries  of Theorem
\ref{teo1_metodo}. This is of course not a complete list and  one can find other choices in \cite{BaCFG}.

\smallskip
A first class of choices, closely related  to the ones   in \cite[Section 4.5]{BaCFG}, is given in the following \rosso{corollary}:
\begin{corollary}\label{ristretto} Suppose that \rosso{$\cK(t)= (\cK_{ij}(t))$} is a \rosso{time--periodic} generalized flow with $\div \cK=0$ and such that
\rosso{$\overline{\la \a, \cK \ra }\not =0$}.
Then it holds
 \begin{equation}
I_{ \a,\g }(y) \leq  \rosso{\frac{1}{4}}
 \frac{ \hat{\s}
 }{
{\overline{ \la  \a , \cK \ra}\,} ^2
}
 (y- y_{\g, \a}) ^2 + o \bigl( (y- y_{\g, \a}) ^2\bigr)
\end{equation}
and
\begin{equation}\label{lavori_hat}
 D_{\a,\g}\geq \frac{ \rosso{{\overline{ \la  \a , \cK \ra}{\,}} ^2}}{\hat{\s}}\,,
\end{equation}
   where
   \begin{equation}\label{sigma_hat}
   \hat{\s}:=  2 \sum _{(i,j)\in E}
   \rosso{
     \overline{\left(\frac{\cK_{ij}^2 }{\cQ_{ij}} \right)}
      }\,.
  \end{equation}
\end{corollary}
\begin{proof}
It is enough to apply Theorem \ref{teo1_metodo} with   $R:=(y_{\a,\g}/\rosso{\overline{\la \a, \cK \ra } }) \cK $
 and $m=0$.
\end{proof}

\rosso{We collect  some comments on the above Corollary \ref{ristretto}:}
\begin{itemize}
\item A possible choice of $\cK$ is given by  $\cK=\overline{\cQ}$ when $\la \overline{ \a}, \overline{\cQ}\ra \not=0$.

\item
When $\g\equiv 0$ and $\a$ is time--independent  \rosso{we have that   $y_{\a, \g}= \overline{ \la \a ,  \cQ \ra}=\la \a, \overline \cQ\ra$}.
 In particular,  by taking $\cK=\overline{\cQ}$ in the above \rosso{Corollary} \ref{ristretto}, \eqref{lavori_hat} becomes  \eqref{zac} valid whenever $\la \a, \overline{\cQ}\ra\not=0$.

\item \rosso{Another possible choice for $\cK$ is given by $\cK_{ij}(t)= \mu_i(t) w_{ij} (t)$, where $\mu_i(t)$ denotes the so--called accompanying distribution, i.e. the invariant distribution for the time--homogeneous Markov chain with time--independent rates $w_{ij}(t)$ ($t$ thought of as frozen). For this second choice  we  also refer to \cite[Section 4.5]{BaCFG}.}

\item \rosso{The property of being a  time periodic generalized flow with zero divergence  is preserved by linear combinations. In particular, one can also  take  $\cK_{ij}=c_1 \overline{\cQ}_{ij}+ c_2  \mu_i(t) w_{ij} (t)$, for any fixed $c_1,c_2\in \bbR$.}

\item Given the  model, one can look for more efficient choices \rosso{of $\cK$}
by using   Schnakenberg's  cycle theory \cite{BFG2,Sch} to build divergence--free flows,  and afterwards by trying to optimize \rrr{among these flows}.  Note that non--trivial divergence--free \rosso{flows} on the graph $(V,E)$ always exist.

\end{itemize}
%
%
We are not going to discuss in detail the possible optimization problems related to the \rosso{last comment above, concerning  Schnakenberg's  cycle theory, since this approach is very model--dependent.} 
 We consider in the next section  just one special case where an argument of this type works \rrr{naturally} (cf. first proof of Theorem \ref{teo2_metodo}).

\smallskip

In \cite{KSP} the authors consider functionals of the form \eqref{gattino} with $\a$ antisymmetric and $\g$ not arbitrary, but of the form
\begin{equation}\label{salvezza}
\g_i (t)= \frac{d}{dt} g_i(t) \qquad \forall i \in V\,,
\end{equation}
for some periodic function $g_i$.  \rosso{The above form \eqref{salvezza} is equivalent to the property}
\begin{equation}\label{integrino}
\overline \g_i =0\qquad \forall i \in V\,.
\end{equation}
In the following result we consider general weights $\a$  and we  weaken condition \eqref{integrino} on $\gamma$.

\begin{corollary}\label{teo1_TUR}  Suppose that the entries of  $\overline\g$ are not all strictly positive, and not all strictly negative. Fix any time--independent probability measure $p=(p_i)_{i \in V}$ on $V$ with $\la p, \overline{\g}\ra=0$.  Recall the constant $C(p)$ defined in \eqref{cipi}.
 Then we have  the
upper bound
\begin{equation}\label{flash1}
I_{\a,\g}(y) \leq\frac{C(p)}{4} \frac{(y-y_{\a,\g})^2}{y_{\a,\g}^2}+ o\left(  (y-y_{\a,\g})^2\right)\,.
\end{equation}
As a consequence we have  \eqref{tanos1}.
\end{corollary}
\begin{proof}
We take
\[ R(t):= \cQ(t) \,, \qquad m(t):= \pi(t) - p\,.
\]
Properties (P1), (P2), (P3) are  satisfied (recall the continuity equation \eqref{cont_eq}).
Due to \eqref{LLN} and \eqref{integrino} we have
\[
\overline{ \la \a  , R  \ra}+\overline{\la \g , m  \ra}=y_{\a,\g} -\la \overline  \g ,p \ra=y_{\a,\g} \,.
\]
Hence, also property (P4) is satisfied.
Note that
\[ R_{ij}(t)-m_i(t) w_{ij}(t)=   \cQ_{ij}(t) -  ( \pi_i(t) - p_i) w_{ij}(t) = p_i w_{ij}(t) \,. \]
By plugging the above identity in \eqref{flash-a} and \eqref{tanos-a} we get
\eqref{flash1}.    \eqref{tanos1} then follows due to \eqref{siluro}.
\end{proof}

\begin{remark}
We would like to point out that the art of finding good bounds is
related to the art of finding good perturbations $(m,R)$ and this is
essentially the art of finding periodic solutions of the
continuity equation in condition (P3). We briefly discuss in Appendix
\ref{aoh} \rrr{two} possible approaches.
\end{remark}

%

%
%

\section{\rosso{Local bounds on $I_{\a,\g}$
 and GTUR's for $Y_{\a, \g}^{(n)}$:  case of antisymmetric   $\a$}}
\label{sec_TUR_anti}
In all this section we will assume, without further mention, that
\[
\begin{cases}
\text{$\a$ is antisymmetric, i.e. $\a_{ij}(t)=-\a_{ji}(t)$}\,,\\
\text{$(i,j)\in E$ if and only if $(j,i)\in E$}\,.
\end{cases}\]
Below we provide two
 general methods to get local quadratic bounds on  $I_{\a,\g}$
   (see Theorems \ref{teo2_metodo} and \ref{teo3_metodo}) and we discuss some corollaries.
We prove  Theorem  \ref{teo2_metodo}  using two approaches. For the first proof, we start with Theorem \ref{teo1_metodo}
and perform an optimization  among flows.
 Hence, Theorem \ref{teo2_metodo} can be seen as
corollary of Theorem \ref{teo1_metodo}. For the second proof, we use the LD rate functional associated with the
empirical current and empirical measure from \cite{BCFG}.
Also for Theorem \ref{teo3_metodo} we provide two alternative derivations. In one derivation we  get  Theorem \ref{teo3_metodo}  from  Theorem \ref{teo2_metodo}. As a consequence, both Theorems  \ref{teo2_metodo}   and \ref{teo3_metodo} follow from Theorem \ref{teo1_metodo}.

\subsection{Preliminaries and Theorem \ref{teo2_metodo}}
In what follows, we call \emph{current} any function $d: V\times V \to \bbR$ which is zero outside $E$ and antisymmetric, i.e. $d_{ij}= - d_{ji}$ $\forall i,j$.
 We  order the elements of $V$ (arbitrarily) and write  $< $ for the order relation. Given $a: V\times V \to \bbR$, we define
\[ \la\la a, d \ra\ra:=  \sum _{(i,j)\in E: i<j}  a_{ij}d _{ij}\,.\]
Note that, when also $a$ is antisymmetric, we have $\la \la a, d\ra\ra=\frac{1}{2}\la a, d\ra$. Finally, we define the divergence of  a current $d$ by
\begin{equation}\label{cdiv}
 {\rm div}_i \,d= \sum _j d_{ij} \,.
 \end{equation}
 We point out  that the divergence of a current is defined  differently from the divergence of a generalized flow (see \eqref{pera}).  This definition guarantees that, if $k$ is a  generalized flow and $d$ is the current $d_{ij}:= k_{ij}-k_{ji}$, then $\la\la d',d\ra\ra=\la d', k\ra$ for any current $d'$.

\smallskip

Due to the antisymmetry of the increments $\a_{ij}$, the LD rate functional $I_{\a, \g}$
admits an alternative variational characterization, in addition to \eqref{variazionale}, in
terms of the empirical current and density \cite{BCFG}, as explained
below.
We consider  the function
 \begin{equation}
 \psi (j,g,a):=\sqrt{g^2+a^2} -\sqrt{j^2+a^2}+ j \bigl[\sinh^{-1}(j/a)-\sinh^{-1}(g/a)\bigr]\,,
 \end{equation}
$\sinh^{-1}(x)$ denoting the hyperbolic arcsinus.  \rosso{Then it holds}
 \begin{equation}\label{variazionale_bis}
I_{\a, \g}(y)= \inf \{ \rosso{I_*(J,\rho)}\,:\, (J,\rho)  \in \cF^{*}_{\a, \g,y}  \}\,,
\end{equation}
where  now
\begin{align}
&\rosso{I_*(J,\rho)} :=\sum_{(i,j)\in E: i<j}  \overline{\Psi \bigl( J_{ij}(t), G_{ij}(t), a_{ij}(t) \bigr)}\,,\label{def_F_bis} \\
& G_{ij}(t):= \rho_i(t)w_{ij}(t)-\rho_j(t) w_{ji}(t)\,,\label{def_G}\\
& a_{ij}(t) := 2\sqrt{ \rho _i(t) \rho _j(t) w_{ij}(t) w_{ji}(t) }\,,  \label{def_a}
\end{align}
 and $\cF^*_{\a,\g,y}$ denotes the family of pairs  $(J,\rho)=\left( J(t), \rho (t)\right)_{t\geq 0}$ such that
\begin{itemize}
\item[(i)]
 $J(t)$ is a time--periodic current, i.e. $J(t)=J(t+\t)$ and $J(t)$ is an antisymmetric function   on $V\times V$ which is zero outside $E$  for each time  $t$;
 \item[(ii)]  $\rho(t)$   is a time--periodic probability measure on $V$;
 \item[(iii)] the continuity equation $\partial _t \rho (t) + \div J(t)=0$ is satisfied \rosso{(cf. \eqref{cdiv})};
  \item[(iv)] \rosso{$y =\overline{ \la\la  \a , J \ra \ra}+\overline{\la \g , \rho \ra}$}.
\end{itemize}

According to  \cite[Theorem 3]{BCFG}, formula \eqref{def_F_bis} with the restrictions (i), (ii) and (iii) is the joint LD rate function for the   empirical  current and measure \ppp{with speed $n \t$}. \rosso{The empirical current $J^{(n)}$ is defined as the measure on $E\times [0,\t)$ given by
$J^{(n)}(i,j,dt):= Q^{(n)} (i,j,dt)- Q^{(n)} (j,i,dt)$ (cf. Section \ref{principiante}). Formula \eqref{def_F_bis}}
 can be deduced directly by contraction starting from the joint LD rate functional for the empirical measure and flow discussed  in Section \ref{principiante}. As in \eqref{rappresento} we have
  \begin{equation}\label{rappresento_anti}
 Y^{(n)}_{\a,\g}=\ppp{\frac{1}{\t}}\sum _{(i,j)\in E: i<j }\int  \a_{ij}(t) J^{(n)} (i,j, dt)+\ppp{\frac{1}{\t}} \sum_i \int  \g_{i}(t) \rho^{(n)} (i ,dt)\,,
 \end{equation}
 thus allowing to derive the LD principle for $Y^{(n)}_{\a,\g} $ from the LD principle of $\left(J^{(n)}, \rho^{(n)}\right)$ by contraction.
 We  point out that the asymptotic pair $(\cJ, \pi)$ belongs to $\cF^*_{\a,\g,y}$ with $y=y_{\a,\g}$ and that it fulfills the identity $I_*(\cJ, \pi)=0$.

\smallskip

As in Section \ref{sec_TUR_Y}, from now on we assume that  $y_{\a, \g}\not =0$.
 \begin{remark}\label{zanna_bianca}
Suppose that the transition  graph  $(V,E)$ has the property that (i)  for each edge in $E$ also the reversed edge belongs to  $E$, (ii) the non--oriented graph obtained from $(V,E)$ by disregarding the edge orientation is a tree.
In this case, being divergence--free, $\overline{\cJ}$  must be zero and, as a consequence, $y_{\a,0}=0$ for currents with time--independent increments. Moreover, reasoning as in the last paragraph of Section \ref{sec_MC}, one can show that the set $\cF_{\a,\g,y}^*$ defined after \eqref{def_a} is empty for $y\not= 0$, thus implying that $I_{\a,\g}(y)=+\infty$ for $y\not =0$ and $I_{\a,\g}(0)=0$.
\end{remark}

As in Section \ref{sec_TUR_Y} we consider pairs  $(J,\rho)$ written as perturbations of the stationary pair $(\cJ,\pi)$ as follows:
\begin{equation}\label{denti}
\begin{cases}
J=\cJ+\frac{y-y_{\a,\g} }{y_{\a,\g}} Z\,,\\
\rho= \pi+\frac{y-y_{\a,\g} }{y_{\a,\g}} m\,.
\end{cases}
\end{equation}
To assure that $(J, \rho) \in \cF^*_{\a,\g, y}$,
the pair $(Z,m)$ must satisfy the following properties:

\smallskip

\begin{itemize}
	\item[(P1$^*$)] $Z= \big(Z(t)\big)_{t\geq 0}$ is a time--periodic current (in particular  $Z(t): V\times V\to \bbR$ is antisymmetric and is  zero outside $E$ for all $t\geq 0$);
	\item[(P2$^*$)] $m=\bigl( m(t) \bigr) _{t\geq 0}$ is time--periodic and $m(t): V \to \bbR$ satisfies $\sum_i m_i(t)=0$ for all $t\geq 0$;
	\item[(P3$^*$)] $\partial _t m + \div Z(t) =0$ \rosso{(cf. \eqref{cdiv})};
	\item[(P4$^*$)] \rosso{$y_{\a,\g} =\overline{ \la\la \a, Z\ra \ra}+\overline{\la \g , m \ra}$;} 	\item[(P5$^*$)]  it holds    $ \pi_i(t) +\frac{y-y_{\a,\g} }{y_{\a,\g}}   m_i(t) \geq 0$ for all $t\geq 0$ and $i \in V$.
\end{itemize}
\rosso{Since $\pi_i(t)>0$ for any $t$, condition (P5$^*$) is satisfied for $y$ near enough to $y_{\a,\g}$. As a consequence, in what follows we disregard condition (P5$^*$).}

\smallskip

\rrr{We can finally state our general method to get local quadratic bounds on $I_{\a,\g}$}:

\begin{theoremA}\label{teo2_metodo}
	For any pair $(Z,m)$ fulfilling  the above properties (P1$^*$),...,(P4$^*$)  the following local quadratic upper  bound holds:
	\begin{equation}\label{flash}
	I_{\a,\g}(y) \leq\frac{1}{2} \frac{(y-y_{\a,\g})^2}{y_{\a,\g}^2} \sum_{(i,j)\in E:i<j} \overline{
	\left(
	\rosso{
	\frac{ \bigl(Z_{ij} - (m_i  w_{ij} -m_j w_{ji}  )\bigr)^2}{\cQ_{ij}  +\cQ_{ji}}
	}
	\right)
	}+ o\left(  (y-y_{\a,\g})^2\right)\,.
	\end{equation}
	In particular, we have the lower bound
	\begin{equation}\label{tanos}
	2 D_{\a, \g}\geq  y_{\a,\g} ^2  \left\{
	\sum_{(i,j)\in E:i<j} \overline{
	\left(
	\rosso{
	\frac{ \bigl(Z_{ij} - (m_i  w_{ij} -m_j w_{ji}  )\bigr)^2}{\cQ_{ij}  +\cQ_{ji}}
	}
	\right)
	}
	\right\}^{-1} \,.
	\end{equation}
\end{theoremA}
We point out that, since \eqref{denti} defines a  bijection $(Z,m)\mapsto (J,\rho)$,  one would get an identity in \eqref{flash} and \eqref{tanos}  by optimizing among $(Z,m)$ in the above theorem.

\subsection{First proof of Theorem \ref{teo2_metodo}}
The proof relies on Theorem \ref{teo1_metodo} and an optimization procedure in the same spirit of  the last comment on  Corollary   \ref{ristretto}.

Let $(Z,m)$ be a pair fulfilling properties  (P1$^*$),...,(P4$^*$) and let  $R'$ be the time--periodic generalized flow given by $R'_{ij}(t):=Z_{ij}(t)/2$. Note that  the  pair $(R',m)$ satisfies properties (P1),...,(P4) in Section \ref{sec_TUR_Y}. We take $R=(R(t))_{t\geq 0}$ as $R(t):= R'(t)+S(t)$, where  $S(t)$ is a  generic  time--periodic symmetric  generalized  flow, i.e. $S_{i,j}(t)=S_{j,i}(t)$ for all $i,j,t$. Since $\div S(t)=0$ and  $\langle \a(t),S(t)\rangle=0$ by the antisymmetry of $\a$, also  the pair $(R,m)$ satisfies conditions (P1),..,(P4) and therefore  Theorem \ref{teo1_metodo} \rrr{applies} to  $(R,m)$.

\rrr{We  optimize the upper bound  \eqref{flash-a}
in Theorem \ref{teo1_metodo} over the symmetric generalized flows $S$.}
For the optimization, the basic computation that we need is the following. We consider some fixed numbers $r_k, a_k, q_k$, $k=1,2$ and  compute
\begin{equation}\label{questa}
\inf_{s\in \mathbb R}\left[\frac{(r_1+s-a_1)^2}{q_1}+\frac{(r_2+s-a_2)^2}{q_2}\right]\,.
\end{equation}
The function is minimized at
$$
s^*=\frac{c_1}{c_1+c_2}(a_1-r_1)+\frac{c_2}{c_1+c_2}(a_2-r_2)\,,
$$
where $c_k:=q_k^{-1}$. The minimal value is given by
\begin{equation}\label{quella}
\frac{\left[(r_1-r_2)-(a_1-a_2)\right]^2}{q_1+q_2}\,.
\end{equation}
Let us come back to the upper bound \eqref{flash-a}
in Theorem \ref{teo1_metodo}. Independently for each pair of edges $(i,j) $ and $(j,i)$, we can evaluate
\[ \inf_{s\in \bbR }\left\{
 \frac{\bigl(R'_{i j}(t) +s  -m_i(t) w_{ij}(t) \bigr)^2}{ \cQ_{ij}(t) }+
  \frac{\bigl(R'_{ji}(t) +s  -m_j(t) w_{ji}(t) \bigr)^2}{ \cQ_{ji}(t) }\right\}\,,
\]
where $s$ has to be thought as the value $S_{ij}(t)=S_{ji}(t)$.
According to \eqref{quella} the above infimum is indeed attained at a suitable value $S_{ij}^*(t)$ and equals
\begin{equation}
\frac{ \left[ \bigl(R'_{i j}(t)-R'_{ji}(t)\bigr) -\bigl (m_i(t) w_{ij}(t) - m_j(t) w_{ji}(t)\bigr)  \right]^2}{  \cQ_{ij}(t)+ \cQ_{ji}(t)}\,.
\end{equation}
As a consequence,  by taking $R(t)= R'(t)+S^*$
the resulting bound \eqref{flash-a} reduces to \eqref{flash} since $Z_{ij}(t):=R'_{ij}(t)-R'_{ji}(t)$.
Finally, \eqref{tanos} follows from \eqref{flash} and \eqref{siluro}.
%

\subsection{Second  proof of Theorem  \ref{teo2_metodo}}
 \rosso{We follow the same arguments} of Theorem \ref{teo1_metodo} but applied to the functional \rosso{\eqref{def_F_bis}. We first} consider the Taylor's expansion up to the second \rosso{order} of the function $\Psi(j,g,a)$ around the point $(x-y,x-y, 2\sqrt{xy})$  with $x,y\geq 0$.
 By writing
\begin{equation}
\begin{cases}
&j= x-y +\d j \,,\\
& g = x-y + \d g \,,\\
& a= 2\sqrt{xy} + \d a\,,
\end{cases}
\end{equation}
 we have \rosso{(after  cumbersome but straightforward computations) that}
 \begin{equation}\label{second_order}
 \begin{split}
\psi (j,g,a) &= \frac{1}{2} \frac{1}{x+y}(j-g)^2 + o \bigl( (\d j )^2\bigr)+ o \bigl( (\d g)^2\bigr) +
o\bigl((\d a )^2 \bigr)\\
& \rosso{= \frac{1}{2} \frac{1}{x+y}(\d j-\d g)^2 + o \bigl( (\d j )^2\bigr)+ o \bigl( (\d g)^2\bigr) +
o\bigl((\d a )^2 \bigr)}\,.
\end{split}
 \end{equation}
 \rosso{By \eqref{def_G} and \eqref{def_a} we can write
 \begin{align}
  J_{ij}(t) & = \cQ_{ij}(t) - \cQ_{ji}(t) + \frac{y-y_{\a,\g}}{y_{\a,\g}} Z_{ij}(t) \,,\\
  G_{ij}(t) & = \cQ_{ij}(t) - \cQ_{ji}(t)  +\frac{y-y_{\a,\g}}{y_{\a,\g}} \bigl[ m_i(t) w_{ij}(t) - m_j(t) w_{ji}(t) \bigr]\,,\\
  a_{ij} (t) &= 2\sqrt{\cQ_{ij}(t) \cQ_{ji}(t) } + \d a_{ij}(t)
 \end{align}
 where $\d a_{ij}(t)=O \bigl( |y-y_{\a,\g}|\bigr)$ (i.e.  $| \d a_{ij}(t)|\leq C  |y-y_{\a,\g}|$ for $y$ near to $y_{\a,\g}$). Due to the above identities,  applying \eqref{second_order}
 with $x= \cQ_{ij}(t)$ and $y=\cQ_{ji}(t)$
 we get
 \begin{multline}
 \psi\bigl( J_{ij}(t), G_{ij}(t), a_{ij}(t) \bigr)=\\\frac{1}{2}\frac{(y-y_{\a,\g})^2}{y_{\a,\g}^2} \frac{  \left[ Z_{ij}(t)- \bigl( m_i(t) w_{ij}(t) - m_j(t) w_{ji}(t) \bigr) \right]^2}{\cQ_{ij}(t) + \cQ_{ji}(t)}+ o\left(  (y-y_{\a,\g})^2\right)\,.
 \end{multline}
 From this equation together with   \eqref{variazionale_bis} and \eqref{def_F_bis} we get \eqref{flash}. Finally, \eqref{tanos} follows from \eqref{flash} by \eqref{siluro}.}

\subsection{Corollaries to Theorem \ref{teo2_metodo}}
Likewise the previous section we have also the following results.

\begin{corollary} \label{ristretto_bis}   Suppose that \rosso{$\cK(t)= \bigl(\cK_{ij}(t)\bigr)$ is a  time--periodic} current with $\div \cK=0$ and such that \rosso{$ \overline{\la\la \a,\cK  \ra\ra} \not=0$}.
Then it holds
 \begin{equation}
I_{ \a,\g }(y) \leq
 \rosso{
 \frac{1}{4}
 \frac{ \tilde{\s}
 }{
\overline{ \la\la  \a , \cK\ra \ra} ^2
}
}
 (y- y_{\g, \a}) ^2 + o \bigl( (y- y_{\g, \a}) ^2\bigr)
\end{equation}
and
\begin{equation}\label{lavori_tilde}
 D_{\a,\g}\geq \frac{ \rosso{\overline{ \la \la \a , \cK\ra  \ra} ^2}}{
 \tilde{\s}}\,,
\end{equation}
   where
   \begin{equation}\label{sigma_tilde}
  \tilde{\s}:= 2\sum _{(i,j)\in E: i<j}
   \rosso{
     \overline{\frac{\cK_{ij}^2}{\cQ_{ij}+ \cQ_{ji}}}
  }\,.
  \end{equation}
\end{corollary}
We point out that Corollary \ref{ristretto_bis} was also obtained in \cite{BaCFG} and it is an immediate consequence of \rosso{Theorem} \ref{teo2_metodo} with
$\rosso{
Z:=\left(y_{\a,\g}/\,\overline{\la\la \a, \cK\ra\ra}\right) \cK}$ and $m=0$.

\smallskip

As in Section \ref{sec_TUR_Y} we can  collect  some comments on the above Corollary \ref{ristretto_bis}:
\begin{itemize}
\item A possible  choice of $\cK$ is given by  $\cK=\overline{\cJ}$ when  $\la \la \overline{ \a}, \overline{ \cJ}\ra \ra \not=0$.
\item
When $\g\equiv 0$ and $\a$ is time--independent  we have that   $y_{\a, \g}= \overline{ \la\la \a ,  \cJ \ra \ra}=\la \la \a, \overline \cJ \ra\ra$.
 In particular,  by taking $\cK=\overline{\cJ}$ in the above Corollary \ref{ristretto_bis}, \eqref{lavori_tilde} becomes  \eqref{zac_anti}
  valid whenever $\la \a, \overline{\cJ}\ra\not=0$.
\item Another possible choice for $\cK$ is given by $\cK_{ij}(t)= \mu_i(t) w_{ij} (t)-\mu_j w_{ji}(t)$, where $\mu_i(t)$ denotes the so--called  accompanying  distribution (cf. Section \ref{sec_TUR_Y}). For this second choice  we  also refer to \cite[Section 4.5]{BaCFG}.
\item The property of being a  time periodic current  with zero divergence  is preserved by linear combinations. In particular, one can also  take  $\cK_{ij}=c_1 \overline{\cJ}_{ij}+\rrr{ c_2  \left( \mu_i(t) w_{ij} (t)-\mu_j(t)w_{ji}(t) \right)}$, for any fixed $c_1,c_2\in \bbR$ (see \cite[Section 4.5]{BaCFG} for further discussions).
\item Given the  model, one can look for more efficient choices of $\cK$
by using   Schnakenberg's  cycle theory \cite{BFG2,Sch} to build divergence--free currents,  and afterwards by trying to optimize among these currents.  We recall that any divergence--free current $\cK$   must be zero     if the graph $(V,E)$ is  a tree  after replacing pairs of oriented edges $(i,j)$ and $(j,i)$ by the unoriented edge $\{i,j\}$. In this case Corollary \ref{ristretto_bis}  becomes empty.
\end{itemize}

\begin{corollary}\label{teo2_TUR}
 Suppose that the entries of  $\overline\g$ are not all strictly positive, and not all strictly negative. Fix any time--independent probability measure $p=(p_i)_{i \in V}$ on $V$ with $\la p, \overline{\g}\ra=0$.  Recall the definition of $C_{\rm a}(p)$ in \eqref{cipia}.
Then we have  the
upper bound
\begin{equation}\label{flash_teo2}
I_{\a,\g}(y) \leq\frac{1}{4
} \frac{(y-y_{\a,\g})^2}{y_{\a,\g}^2} C_{\rm a}(p)
 + o\left(  (y-y_{\a,\g})^2\right)\,.
\end{equation}
As a consequence we have \eqref{tanos_teo2}.
\end{corollary}
The above result follows from Theorem  \ref{teo2_metodo} by taking $Z(t):= \cJ(t)$ and  $m(t):= \pi(t) - p$.


\begin{remark}\label{mushadara}
Note that both \rosso{Corollary  \ref{ristretto_bis} and \ref{teo2_TUR}  could be derived} respectively from Corollary  \rosso{\ref{ristretto} and \ref{teo1_TUR}} by an optimization over symmetric generalized flows \rosso{as} in the first proof of \rosso{Theorem} \ref{teo2_metodo}. In particular, in the case of an antisymmetric $\a$, the bounds discussed in this section are better \rosso{than}  the corresponding ones discussed in the previous \rosso{section}, since they are obtained by an optimization \rosso{procedure. We recall that we have proved in a direct way this issue (cf. Prop. \ref{fiordaliso} and its proof in Appendix \ref{marvel})}.
\end{remark}

\subsection{Theorem \ref{teo3_metodo} and its corollaries}\label{bingo}


In Theorem \ref{teo3_metodo} below we  present   another general method to produce quadratic bounds on the  LD
rate functional $I_{\a,\g}$. We provide two simple derivations of this theorem.  The first one is inspired by the approach followed  in \cite[Section 4.1]{BaCFG}. The second one,  based  on Theorem \ref{teo2_metodo}, shows indeed that the bounds provided by Theorem \ref{teo2_metodo} are better than the ones provided by Theorem  \ref{teo3_metodo} (see Remark \ref{acciughe} below). Nevertheless, the interest to Theorem \ref{teo3_metodo} comes from the fact that it allows (see the corollaries below) to get GTUR's with constants resembling in their form to the average entropy production rate $\s$.


\begin{theoremA}\label{teo3_metodo}
	For any pair $(Z,m)$ fulfilling   properties (P1$^*$),...,(P4$^*$)  the following local quadratic upper  bound holds:
	\begin{equation}\label{flash3}
	I_{\a,\g}(y) \leq\frac{1}{4} \frac{(y-y_{\a,\g})^2}{y_{\a,\g}^2} \sum_{(i,j)\in E:i<j}
	\overline{
	\left(
	\frac{ \bigl(Z_{ij} - (m_i  w_{ij} -m_j w_{ji}  )\bigr)^2}{\cJ_{ij}} \ln \frac{\cQ_{ij}}{\cQ_{ji}}
	\right)}
	+ o\left(  (y-y_{\a,\g})^2\right)\,.
	\end{equation}
	In particular, we have the lower bound
	\begin{equation}\label{tanos3}
	 D_{\a, \g}\geq  y_{\a,\g} ^2  \left\{ \sum_{(i,j)\in E:i<j}
	\overline{\left(
		\frac{ \bigl(Z_{ij} - (m_i  w_{ij} -m_j w_{ji}  )\bigr)^2}{\cJ_{ij}} \ln \frac{\cQ_{ij}}{\cQ_{ji}}
	\right)}
	\right\}^{-1} \,.
	\end{equation}
\end{theoremA}
\begin{proof}[First proof]
We have (recall \eqref{variazionale_bis} and \eqref{def_G})
\begin{equation}\label{silenzioso}
I_{\a, \g} (y) \leq  \rosso{I_*(J,\rho)} \leq \frac{1}{4} \sum _{(i,j)\in E: i<j } \frac{1}{\t}\int_0^\t \frac{[ J_{ij}(t)-G_{ij}(t)]^2 }{G_{ij}(t)} \ln \frac{\rho_i(t) w_{ij}(t)}{ \rho_j(t) w_{ji}(t)}dt\,,
\end{equation}
for any pair $(J,\rho)$ in $\cF^*_{\a, \g, y}$.
The second bound in \eqref{silenzioso} follows from Eq. (12) in \cite{GHPE}, \rosso{implying that
\[ \Psi \bigl( J_{ij}(t), G_{ij}(t), a_{ij}(t) \bigr) \leq \frac{1}{4} \frac{[ J_{ij}(t)-G_{ij}(t)]^2 }{G_{ij}(t)} \ln \frac{\rho_i(t) w_{ij}(t)}{ \rho_j(t) w_{ji}(t)}\,.
\]}
We take the pair $(J,\rho)$ as in \eqref{denti}.  Then, for $y $ close to $y_{\a, \g}$, we have that
$(J, \rho)\in \cF^*_{\a,\g,y}$ and therefore we can apply \eqref{silenzioso} to $(J,\rho)$. The thesis then follows by a Taylor's expansion of the r.h.s. of \eqref{silenzioso} for $y$ close to  $y_{ \a,\g}$, since
\begin{align}
&J_{ij}(t)- G_{ij}(t)=\frac{y-y_{\a,\g} }{y_{\a,\g}} \left( Z_{ij}(t)- [m_i(t) w_{ij}(t)-m_j(t) w_{ji}(t)]\right)\,,\label{uffa1}\\
& G_{ij}(t)=\cJ_{ij}(t) + o(1)\,,\label{uffa2} \\
& \frac{\rho_i(t) w_{ij}(t)}{ \rho_j(t) w_{ji}(t)} = \frac{ \cQ_{ij}(t)}{\cQ_{ji}(t)} +o(1)\,.\label{uffa3}
\end{align}
\end{proof}

\begin{proof}[Second proof]
The bound \eqref{flash3} is an immediate consequence of the bound \eqref{flash} in Theorem \ref{teo2_metodo} and the \blu{general}  inequality \blu{(cf. \cite[Eq. (29)]{BaCFG})}
 \begin{equation}\label{elementare}
  (x-y) \ln \frac{x}{y} \geq \frac{2 (x-y)^2}{x+y}\,,\qquad x,y>0\,.
  \end{equation}
\blu{Indeed, from  the above inequality one gets that $(\cJ_{ij})^{-1}\ln (\cQ_{ij}/\cQ_{ji})\geq 2 (\cQ_{ij}+\cQ_{ji})^{-1}$.}
Finally, \eqref{tanos3} follows from \eqref{flash3} and \eqref{siluro}.
\end{proof}

\begin{remark}\label{acciughe}
Due to \eqref{elementare}, the r.h.s. of \eqref{tanos} in Theorem \ref{teo2_metodo} is  lower bounded  by  the r.h.s. of \eqref{tanos3} in Theorem \ref{teo3_metodo}. In particular, the bounds obtained by Theorem \ref{teo2_metodo} are better than the corresponding  bounds obtained by Theorem \ref{teo3_metodo}.
\end{remark}

\begin{corollary} \label{ristretto_tris}   Suppose that \rosso{$\cK(t)= \bigl(\cK_{ij}(t)\bigr)$ is a  time--periodic} current with $\div \cK=0$ and such that \rosso{$ \overline{\la\la \a,\cK  \ra\ra} \not=0$}.
Then it holds
 \begin{equation}\label{mazinga}
I_{ \a,\g }(y) \leq
 \rosso{
 \frac{1}{4}
 \frac{ \s^*
 }{
\overline{ \la\la  \a , \cK\ra \ra} ^2
}
}
 (y- y_{\g, \a}) ^2 
\end{equation}
and
\begin{equation}\label{lavori_star}
 D_{\a,\g}\geq \frac{ \rosso{\overline{ \la \la \a , \cK\ra  \ra} ^2}}{
 {\s}^* }\,,
\end{equation}
   where
   \begin{equation}\label{sigma_star}
  \s^*:= \sum_{(i,j)\in E: i<j}  \overline{\left(\frac{ \cK_{ij}^2}{\cJ_{ij}} \ln \frac{\cQ_{ij}}{\cQ_{ji}}\right)} \,.
    \end{equation}
\end{corollary}
\begin{proof}
We apply    Theorem \ref{teo3_metodo} with a slight improvement, by  taking $m:=0$ and
$Z:=\left(y_{\a,\g}/\,\overline{\la\la \a, \cK\ra\ra}\right) \cK$. Theorem \ref{teo3_metodo} would imply the thesis, with the exception that the bound \eqref{mazinga} would be only local. On the other hand, since $m=0$, the error terms $o(1)$ in \eqref{uffa2} and \eqref{uffa3} are simply zero and  the first proof of Theorem \ref{teo3_metodo} gives that the local bound \eqref{flash3} is in this case a global bound.
\end{proof}
 We point out that Corollary \ref{ristretto_tris} was also obtained in \cite{BaCFG}.
 Moreover, we observe that  \eqref{liberazione} follows from Corollary \ref{ristretto_tris} by taking $\cK:=\overline{\cJ}$.  Finally, by Remark \ref{acciughe} we also get that
$ \s^*\geq \tilde{\s}$, where
  the constant $\tilde{\sigma}$ is defined as  in \eqref{sigma_tilde}.

\smallskip

\begin{corollary}\label{gnocchi1}
Suppose that the entries of  $\overline\g$ are not all strictly positive, and not all strictly negative. Fix any time--independent probability measure $p=(p_i)_{i \in V}$ on $V$ with $\la p, \overline{\g}\ra=0$.
 Recall the constant $C^*_{\rm a}(p)$ defined in \eqref{cipia_star}.
 Then we have  the
upper bound
\begin{equation}\label{flash100}
I_{\a,\g}(y) \leq\frac{1}{4} \frac{(y-y_{\a,\g})^2}{y_{\a,\g}^2}
C^*_{\rm a}(p)
+ o\left(  (y-y_{\a,\g})^2\right)\,.
\end{equation}
As a consequence  we have   \eqref{tanos100}.
\end{corollary}
The above corollary follows from Theorem \ref{teo3_metodo} by taking $Z(t):= \cJ(t)$ and  $m(t):= \pi(t) - p$, as in Corollary \ref{teo2_TUR}.   \eqref{tanos100} corresponds to \cite[Eq.~(14)]{KSP}.
We point out that, by Remark \ref{acciughe},  we  get that
$C^*_{\rm a}(p) \geq  C_{\rm a} (p) $, where
  the constant $C_{\rm a}(p)$ is defined as  in \eqref{cipia}.

\section{\eqref{gtur_exp} and its extensions}\label{sec_proes}
In this section we generalize the results from \cite{PvdB}  to general protocols that
can be time-asymmetric. Our  GTUR contains the rate $\sigma_{\rm naive}$ that
becomes the average entropy production rate $\s$ for the case of time-symmetric protocols,
as explained in Section  \ref{listone}.

 We denote by   $\Theta_\t$ the set of all  possible paths  of the Markov chain  up to time $\t$  ($\Theta_\t$ is given by the  piecewise-constant paths $\G: [0, \t]\to V$). Note that $\t$ is both the period and the length of the paths.
 We write $\mathcal{R}_\tau: \Theta_\t\to \Theta_\t $ for the time-reflection  around $\t/2$ and we denote by $P$  the probability measure on $\Theta_\t$ given by the law of the random path $ ( X(t))_{0\leq t \leq \tau}$ when the Markov chain  has initial distribution  $\pi(0)$.
  \ppp{Similarly to  \cite{BCFG} we introduce the average  entropy flow from naive reversal defined as the entropy $H[ P\,|\, P\circ \mathcal{R}_\tau] $ of $P$ w.r.t. $ P\circ \mathcal{R}_\t $  (note that $\cR_\t=\cR_\t^{-1}$), i.e.
\begin{equation}\label{rugby}
 H[ P\,|\, P\circ \mathcal{R}_\tau]= \int _{\Theta_\t} P(d\Gamma) \ln \frac{d P}{d(  P\circ R_\t  )}(\Gamma)\,,
\end{equation}
where $P\circ \cR_\t (A):= P(  \cR_\t (A))$.
One gets  (cf.~\cite[Section 4]{BCFG})
\begin{equation}\label{placche}
 H[ P\,|\, P\circ \mathcal{R}_\tau]  =\tau \sigma_{\rm naive}\,,
\end{equation}
where
$ \sigma_{\rm naive} $   is given by \eqref{sigmanaive}.
}

\smallskip
Given a function $F: \Theta_\t\to \mathbb{R}$, we define    the empirical functional $Y^{(n)}_F$  as
\begin{equation}\label{shangai}
Y^{(n)}_F :=\frac{1}{n} \sum _{j=0}^{n-1} F \left( (X_{j\tau+s})_{0\leq s \leq \tau} \right)\,.\end{equation}
We point out that  the empirical functional $Y^{(n)}_{\a,\g}$ given in \eqref{gattino}  can be written as  $Y^{(n)}_{\a,\g}=Y^{(n)}_F $ by defining $F$ as
\begin{equation}\label{chloe}
F(\G):= \frac{1}{\t} \sum _{\substack{ t\in (0, \tau]:\\ \G(t-)\not = \G(t+) }  } \a_{\G(t-), \G(t+) }(t)+ \frac{1}{\t}\int_0 ^{ \tau} \g_{\G(t) }(t) dt \,.
\end{equation}
  We remark that  $Y^{(n)}_{\a,\g}$ is a linear functional  of the empirical flow and density, while the empirical functional  $Y^{(n)}_F$ in \eqref{shangai} is more general.

\begin{theoremA}\label{teo_proes_extended} \emph{[\cite{PvdB} revisited]}
Let  $F: \Theta_\t\to \mathbb{R}$ be antisymmetric, i.e.  $F= -F \circ \mathcal{R_\tau}$.
Then, as $n \to \infty$,  $Y^{(n)}_F $ satisfies an LDP with speed $n $. \rrr{Calling $I_F$
the associated
LD rate function, calling
  $y_F$ the asymptotic value of $Y^{(n)}_F $ and assuming  $y_F\not=0$,} it  holds
 \begin{equation}\label{prosit}
  I_F''(y_F) \leq \ppp{\frac{1}{2  y_F^2} ( e^{ \t \s_{\rm naive}}-1)}\,.
 \end{equation}
\ppp{As a consequence},   one has the GTUR
\begin{equation}\label{singapore}
\rrr{
\frac{D_F }{ y_F^2 }\geq  \frac{\ppp{1}}{  e ^{\tau \s_{\rm naive}}-1}
}\,,
\end{equation}
where $D_F$ is the asymptotic diffusion coefficient  given by
\begin{equation}\label{ciliegia}
2 D_F:= \lim _{n \to \infty} n  {\rm Var} \Big(Y_F^{(n)}\Big)\,.
\end{equation}
\end{theoremA}

\ppp{We note that $y_F=E[F]$ and $2 D_F={\rm Var} (F)$, where the expectation and the  variance are  computed w.r.t. $P$.}

We stress that  Theorem \ref{teo_proes_extended} holds for any protocol, but it is restricted to antisymmetric functionals $F$ as in \cite{PvdB}.
 In Appendix \ref{doraiaki1} we give for completeness the  derivation of Theorem \ref{teo_proes_extended}.  \blu{This proof follows the main steps of the one in \cite{PvdB}, while some mathematical structures are investigated more carefully.  }

%

 In order to apply Theorem \ref{teo_proes_extended} to the functional  $Y^{(n)}_{\a,\g}=Y^{(n)}_F $, with $F$ defined in \eqref{chloe},
we need that $F$ is antisymmetric and this holds whenever
\begin{equation}\label{harry_danger}
\begin{cases}
\a_{i,j}(t)= -\a_{j,i}(\t-t) \,, \\
\rosso{\g_i (t) = - \g_i(\t-t)}\,,
\end{cases}
\end{equation}
for all $i,j\in V$ and all $t\in [0,\t]$. If the weights are time independent, then \eqref{harry_danger} reduces to the fact that $\a$ is antisymmetric (i.e. $\a_{i,j}= \a_{j,i}$) and $\g\equiv 0$.
\ppp{Let us finally explain how to get \eqref{gtur_exp}. By \eqref{pizza78} and \eqref{ciliegia} we have
\begin{equation}
2 D_{\a,\g}  = \lim _{n \to \infty} n \tau \text{Var}\Big(Y^{(n)} _{ \a,\g}\Big)=
 \lim _{n \to \infty} n\t  {\rm Var} \Big(Y_F^{(n)}\Big)=
2 \t D_F\,,
\end{equation}
while
\begin{equation}
y_{\a,\g}= y_F\,.
\end{equation}
As a consequence,  we get that $D_{\a,\g}/y_{\a,\g}^2= \t D_F/y_F^2$. As a byproduct with \eqref{singapore}, we get the desired  \eqref{gtur_exp}.
}


\appendix

\section{Proof of Propositions \ref{girasole}, \ref{fiordaliso}, \ref{violetta} and Remark \ref{ventoso}}\label{marvel}
\subsection{Proof of Proposition \ref{girasole}}
The universal rate in \eqref{tanos1} can be written as $C(p)=\sum_i p_i^2 X_i$, where $X_i=1/A_i$,
$A_i$  is defined in \eqref{aiutino}, and $A_i>0$.
By \eqref{tanos1} we have that
\begin{equation}\label{olivia}
\frac{ D_{\a,\g}}{y_{\a,\g}  ^2} \geq \frac{1}{C_\star} \,,
\end{equation}
where
 $C_\star$ is the infimum of $C(p)$ as  $p=(p_i)_{i\in V}$ varies  among the probability measures on $V$ with $\la \overline\gamma,p\ra=0$.
 Below we show that  the convex  function $p\mapsto C(p)$, defined on the set of  probability measures  with $\la \overline\gamma,p\ra=0$,
 has exactly one extremal point, hence this extremal point must be  the minimum point.

 By the Lagrange's multipliers method, we look to the  extremal points of the function
 \[
f(p)=\sum_{i}p_{i}^{2}X_{i}-a\bigl(\sum_{i}p_{i}-1\bigr)-b\bigl(\sum_{i}p_{i}\overline{\gamma}_{i}\bigr)\,,
\]
 $a,b$ being the multipliers.  The extremal point satisfies  $2p_{i}^{\star}X_{i}-a-b\overline{\gamma}_{i}=0$ for all $ i \in V$, i.e.
\[
p_{i}^{\star}= \frac{a+b\overline{\gamma}_{i}}{2X_{i}} =\frac{a A_i + b A_i \overline\g_i }{2} \qquad \forall i \in V \,.
\]
The constants $a,b$ are fixed by imposing that $\sum_i p^\star_i=1$ and $\la \overline \g, p^\star\ra=0$.
 This is equivalent to the system
 \[
\begin{cases}
aA+bB=2\\
aB+bC=0
\end{cases}
\]
with $A:= \sum _i A_i$, $B:=\sum _i A_i \overline\g_i $ and $C:= \sum _i A_i  \overline{\g}_i^2$.

\smallskip

We point out that by Cauchy--Schwarz inequality  we have
\[  B^2 =\Big(\sum _i A_i \overline\g_i\Big)^2=\Big(\sum_i \sqrt{A_i }(\sqrt{A_i}\overline{\g}_i)  \Big)^2  \leq \Big(\sum_i  A_i\Big)\Big (\sum_i   A_i {\overline\g}_i^2    \Big)=AC\,.
\]
Moreover, the above bound becomes an identity if and only if the vectors $(A_i)$ and $(\sqrt{A_i}\overline{\g}_i)$ are proportional. This condition  is fulfilled in the case given by Item (i) in Proposition \ref{girasole}  since $\overline \g=0$, but not in the case given by Item (ii) in Proposition \ref{girasole}, since $A_i>0$ for all $i$ while $\overline \g\not =0$ has neither all  entries negative nor all  entries positive. Hence, for Item (ii) we have $AC\not =B^2$.

\smallskip

 If $AC\neq B^{2}$, then
the solution of the system is given by $ a=\frac{2C}{AC-B^{2}}$ and $b=\frac{-2B}{AC-B^{2}}$, thus leading to
\begin{align*}
C_\star=C(p^\star)=\sum_{i}\left(p_{i}^{\star}\right)^{2}X_{i} & =\sum_{i}p_{i}^{\star}\left(p_{i}^{\star}X_{i}\right)=\sum_{i}p_{i}^{\star}\left(\frac{a+b\overline{\gamma}_{i}}{2}\right)\\
 & =\frac{a}{2}+\frac{b}{2}\left(\sum_{i}p_{i}^{\star}\overline{\gamma}_{i}\right)=\frac{a}{2}+\frac{b}{2}\sum_{i}\overline{\gamma}_{i}\left(\frac{a A_i +bA_i \overline{\gamma}_{i}}{2} \right)\\
 & =\frac{a}{2}+\frac{ab}{4}B+\frac{b^{2}}{4}C=\frac{C}{AC-B^{2}}.
\end{align*}
This concludes the proof of Item (ii)  in Proposition \ref{girasole} by \eqref{olivia} and by the above observation that  $AC-B^{2}>0$.

\smallskip
For the case corresponding to Item (i) of Proposition \ref{girasole}  with $\overline{\gamma}=0$, the multiplier $b$ can be neglected and $aA=2$. Hence $p_i^\star= A_i/ \sum_j A_j$, which leads to the identity
\[
C_\star=C(p^\star)=\sum_{i}\left(p_{i}^{\ast}\right)^{2}X_{i}=[\sum _i A_i ]^{-1}\,,
\]

\subsection{Proof of Remark \ref{ventoso}}
Since $\pi_i$ is time--independent,  the statement in Remark \ref{ventoso} is equivalent to the inequality
\begin{equation}\label{blowing}
2\sum_i A_i =\sum _{i}\pi_i  \Big[
 \sum_{ j:(i,j)\in E} \overline{w}_{ij} \Big]^{-1}
 \geq \Big[
\sum_{(i,j)\in E} (\overline{\cQ}_{ij})^2\overline{(1/\cQ_{ij})} \Big]^{-1}=2/\hat{\s}\,.
\end{equation}
Recall that, given a positive random variable $Y$, it holds $\bbE[1/Y]\geq 1/\bbE[Y]$ by Jensen's inequality. We apply this inequality twice. As a first application we get
$\overline{(1/\cQ_{ij})} \geq 1/\overline{\cQ}_{ij}$. This implies that
\begin{equation}\label{bravo_lori}
\Big[
\sum_{(i,j)\in E} (\overline{\cQ}_{ij})^2\overline{(1/\cQ_{ij})} \Big]^{-1}\leq \Big[
\sum_{(i,j)\in E} \overline{\cQ}_{ij}) \Big]^{-1}\,.
\end{equation}
As a second application we get
\begin{equation}\label{bravo_pierpi}
\sum _{i}\pi_i  \Big[
 \sum_{ j:(i,j)\in E} \overline{w}_{ij} \Big]^{-1}
 \geq \Big[  \sum _{i} \pi_i
  \sum_{ j:(i,j)\in E} \overline{w}_{ij}
  \Big]^{-1}=\Big[
\sum_{(i,j)\in E} \overline{\cQ}_{ij}) \Big]^{-1}\,.
\end{equation}
\eqref{blowing} is then a byproduct of \eqref{bravo_lori} and \eqref{bravo_pierpi}.
\subsection{Proof of Proposition \ref{fiordaliso}} The last statement in Proposition \ref{fiordaliso} is an immediate consequence of the bounds $C(p) \geq C_a(p)$ and $C_a^*(p)\geq C_a(p)$, on which we focus.  The bound  $C_a^*(p) \geq C_a(p)$  follows from Remark \ref{acciughe} as discussed after Corollary \ref{gnocchi1}.
 Let us prove that $C(p) \geq C_a(p)$. Given $x,y\geq0$ and $X,Y>0$,
 we have
\[
\frac{(x-y)^2}{X+Y}\leq \frac{x^2+y^2}{X+Y} =  \frac{x^2}{X+Y} + \frac{y^2}{X+Y} \leq
 \frac{x^2}{X} + \frac{y^2}{Y}\,.
 \] The above bound implies \begin{equation}\label{cubone}
\begin{split}
\sum _{(i,j)\in E} \frac{ \bigl(p_i w_{ij}(t) -p_j w_{ji}(t)\bigr)^2}{\cQ_{ij}(t)  + \cQ_{ji}(t) }&  \leq
\sum _{(i,j)\in E} \Big[ \frac{
 \bigl(p_i w_{ij}(t)\bigr)^2 }{\cQ_{ij}(t)  }+\frac{
 \bigl( p_j w_{ji}(t)\bigr)^2}{ \cQ_{ji}(t) }\Big]\\
 &= 2\sum _{(i,j)\in E}  \frac{
 \bigl(p_i w_{ij}(t)\bigr)^2 }{\cQ_{ij}(t)  }\,.
 \end{split}
\end{equation}
By taking the time average on $[0,\t]$ in \eqref{cubone}, we conclude that $C_a(p)\leq C(p)$.
\subsection{Proof of Proposition \ref{violetta}}
The bound $\s^* \geq \tilde{\s}$ has been derived in \cite{BaCFG} and follows also from Remark  \ref{acciughe}. The bound $\hat{\s}\geq \tilde{\s}$ can be derived as follows:
\begin{equation}
\begin{split}
\tilde{\sigma}& = \sum _{(i,j)\in E}  (\overline{\cQ}_{ij}-\overline{\cQ}_{ji})^2 \overline{
\frac{1}{ \cQ_{ij}+ \cQ_{ji} } }\\
&
\leq   \sum _{(i,j)\in E}  (\overline{\cQ}_{ij}^2+\overline{\cQ}_{ji}^2) \overline{
\frac{1}{ \cQ_{ij}+ \cQ_{ji} } }
\leq
 \sum _{(i,j)\in E}  \overline{\cQ}_{ij}^2 \overline{
\frac{1}{ \cQ_{ij} } }+ \sum _{(i,j)\in E}  \overline{\cQ}_{ji}^2 \overline{
\frac{1}{ \cQ_{ji} } }
=  \hat{\sigma}
\,.\end{split}
\end{equation}


\section{Periodic solutions of the continuity equation}
\label{aoh}
In this short appendix we illustrate two possible approaches to find
good perturbations $(m,R)$ in Section \ref{sec_TUR_Y}. We present just
the general ideas since a complete development would be long and model--dependent.
\subsection{Time dependent Schnakenberg theory}
\blu{We consider a cycle
\[C=(i_1,i_2,\dots ,i_N, i_1)\] of the transition graph  $(V,E)$ and
look for pairs $(R,m)$ satisfying properties (P1), (P2), (P3)  in Section \ref{sec_TUR_Y},
just restricted to this cycle. }
 Since we have a
one dimensional ring this is relatively easy. The continuity equation
reduces to
\begin{equation}
\label{erpupone}
\dot m_{i_k}=R_{i_{k-1}i_k}-R_{i_k i_{k+1}}\,, \qquad k=1,\dots ,N\,,
\end{equation}
where the sums \blu{$k\pm 1$}  are modulo $N$. The general solution is
therefore given by
$$
\left\{
\begin{array}{l}
m_{i_k}(t)=M_k+\hat \alpha_{k-1}(t)-
\hat\alpha_{k}(t)\,,\\
R_{i_ki_{k+1}}(t) =\alpha_k(t)\,,
\end{array}
\right.
$$
where $\alpha_k$, $k=1,2,\dots ,N$, are arbitrary  time--periodic functions  such that $\int_0^{\tau}\blu{\alpha_k(t)}dt$ does not
depend on \blu{$k$}. The functions \blu{$\hat \alpha_k$} are the corresponding
primitives \red{of $\a_k$} and $M_k$ are arbitrary numbers such that $\sum_{k=1}^N
M_k=0$.

\blu{A special degenerate case is obtained as follows. Consider two particles, performing time--periodic deterministic trajectories  on the cycle $C$.
Call $m$ the difference of the
empirical densities
associated to the trajectory of the first and of the second particle, respectively. Similarly call $R$ the difference of the empirical flows. Then the pair $(m,R)$ satisfies properties (P1), (P2), (P3).}

%
%
Once obtained  solutions on elementary cycles, a trial pair $(m,R)$ \blu{satisfying properties (P1), (P2), (P3)} for the transition graph $(V,E)$
can be obtained as a combination of them.   The classic Schnakenberg
theory allows to construct divergence free flows using cycles. This
approach in a sense is a time--dependent version of this theory, giving
solutions of the continuity equation using the cycle decomposition.
\subsection{Perturbations from Markov models}
Another possible approach that can be useful in specific situations is
obtained by the following observation. Consider a Markov chain with
periodic rates $\tilde w$. If we call $\tilde \pi$ its invariant time
periodic distribution and $\tilde{\mathcal
Q}_{ij}=\tilde\pi_i(t)\tilde w_{ij}(t)$ the corresponding asymptotic
flow we have that $\tilde \pi$ and $\tilde{\mathcal Q}$ are related by
the continuity equation. We can therefore fix the pair $(m,R)$ by
$m_i(t)=M_i-\tilde \pi_i(t)$ and $R_{ij}(t)=\tilde{\mathcal Q
}_{ij}(t)$, where the arbitrary  numbers $M_i$ satisfy the condition
\blu{$\sum_{i}M_i=1$}.
This special way of proceeding can be useful in specific cases where
there is a simple and natural periodic chain to be introduced.

\smallskip

For both approaches   we just discussed the constraints given by $(P1),
(P2)$ and $(P3)$ in Section \ref{sec_TUR_Y}.
To really implement the methods it is necessary to satisfy also the
additional constraint $(P4)$ in Section \ref{sec_TUR_Y}. This further
restriction has to be imposed on the perturbations  discussed above.

\section{Derivation of  Theorem \ref{teo_proes_extended}}\label{doraiaki1}
%
\ppp{We use the same notation introduced in Section \ref{sec_proes}.
The GTUR \eqref{singapore} is an immediate consequence of \eqref{prosit}  and the identity  $ 2 D_F= 1/I''_F(y_F)$. We now explain how to derive \eqref{prosit} .
}

 Recall that $\Theta_\t$ is the family of   piecewise constant paths $\G: [0, \t]\to V$. We denote by $\cP(\Theta_\t)$ the set of  probability measures on $\Theta_\t$.  The expectation w.r.t. $P$ will be denoted by $E[\cdot]$.

We first focus on
 the empirical object
\begin{equation}\label{imperfetto}
Q^{(n)}: = \frac{1}{n} \sum_{j =0}^{n-1} \delta_{ (X_{j\tau +s})_{0\leq s\leq \tau}} \in \mathcal{P}(\Theta_\t)\,.
\end{equation}
Note that $Q^{(n)}$ is the empirical measure of the Markov chain $(W_k)_{k\geq 0}$ on $\Theta_\t$, where $W_k:=(X_{k\tau +s})_{0\leq s\leq \tau}$.
 We point out  that in this Appendix $Q^{(n)}$ is defined as in \eqref{imperfetto} in order to make the notation closer to the one in \cite{PvdB}, in particular $Q^{(n)}$ is not the empirical flow as in the rest of the file (cf. Section \ref{rappresento}).

\medskip

The link with the empirical functional \eqref{shangai} is given by the identity
\begin{equation}
Y_n^{(F)}= \int_{\Theta_\t}  Q^{(n)}(d \G) F(\G) \,.
\end{equation}

 To have  $I(Q)<+\infty$ we need that
\begin{equation}\label{condizione0}
Q(\G_0=i)= Q(\G_\tau=i) \qquad \forall i\in V\,.
\end{equation}
This follows from the fact that $Q^{(n)} ( \G_0=i)= Q^{(n)}(\G_\tau=i)+O(1/n)$ (simply, the
final value  of $(X_{j\tau +s})_{0\leq s\leq \tau}$ equals  the initial value of $(X_{(j+1)\tau +s})_{0\leq s\leq \tau}$).

 As discussed in Subsection \ref{gabriel}  $ Q^{(n)} $ fulfills an LDP with speed $n$ and the associated  LD rate functional $I$ satisfies the inequality
\begin{equation}\label{forma2}
I(Q)\leq H(Q|P)
\end{equation}
for any $Q \in \cP(\Theta_\t) $ satisfying \eqref{condizione0}.
 By the contraction principle \cite{dH} we get that $Y_n^{(F)}$ satisfies an LDP with speed $n$, whose LD rate functional $I_F$ is given by
\begin{equation}\label{rappresentazione}
I_F(y)= \inf \Big\{ I(Q)\,:\, Q\in \cP(\Theta_\t)\,,\; \int_{\Theta_\t}  Q(d \G) F(\G)=y\Big\}\,.
\end{equation}
By combining \eqref{forma2} and \eqref{rappresentazione} we have
\begin{equation}\label{lady_bug}
I_F(y)\leq H(Q|P) \qquad  \forall Q \in \cP(\Theta_\t) \text{ fulfilling } \eqref{condizione0} \text{ and } \int_{\Theta_\t}  Q(d \G) F(\G)=y\,.
\end{equation}
\medskip

We apply \eqref{lady_bug} with some special $Q=Q^y$ that we take  absolutely continuous w.r.t. $P$. Since $y_F\not=0$, for some function $G$ we can write $Q^y$ as
\begin{equation}\label{falco}
\frac{d Q^y}{dP } = 1+ \frac{ y -y_F}{y_F}(1-G)\,.
\end{equation}
Due to \eqref{falco}, the properties $Q^y\in \cP(\Theta_\t)$, $\int_{\Theta_\t}  Q^y(d \G) F(\G)=y$  and \eqref{condizione0} are satisfied if and only if
\begin{equation}\label{trinita}
E[G]=1\,, \qquad E[FG]=0  \;\;\text{ and } \;\; E [ G \mathds{1}_{ \G_0=i}] = E [ G \mathds{1}_{ \G_\tau=i}]\;\; \forall i \in V\,.
\end{equation}
We claim that, using that $F\circ \cR_\t= - F$,  the last two conditions on \eqref{trinita} are always satisfied
if
\begin{equation}\label{usignolo}
\frac{G}{G\circ \mathcal{R_\t}}= \frac{ d P\circ \mathcal{R_\t}}{dP}\,,
\end{equation}
where $P\circ \mathcal{R_\t}$ is the probability on $\Theta_\t$ defined as $P\circ \mathcal{R_\t}(A):= P(  \mathcal{R_\t}(A) )$ for $A\subset \Theta_\t$ measurable.

Let us derive the  claim. Assuming \eqref{usignolo}, we can write
\begin{equation}
\begin{split}
 E\bigl[ GF\bigr]& =-  E\bigl[ G(F\circ \cR_\t)\bigr]= - \int _{\Theta_\t} P\circ \mathcal{R_\t}(d\G) (G\circ \cR_\t)(\G) F(\G)\\
 &=-  E\Big[  \frac{ d P\circ \mathcal{R_\t}}{dP} (G\circ \cR_\t) F\Big]=-E[GF]\,,
\end{split}
\end{equation}
thus implying that $E[GF]=0$ (note that \eqref{usignolo} has been used to get the last identity). Similarly  one can derive \ppp{the last condition of \eqref{trinita}} from \eqref{usignolo}.

One possible choice for \eqref{usignolo} satisfying the constrain $E[G]=1$  is
\begin{equation}\label{frontiera}
G = \frac{(1+ e^{Z})^{-1}}{ E\left[(1+ e^{Z})^{-1}\right]}\,,\qquad e^{-Z}= \frac{ d P\circ \mathcal{R_\t}}{dP}\,.
\end{equation}
Note that $E[Z]=\t \s_{\rm naive}$ \ppp{(cf. \eqref{rugby} and \eqref{placche})}.

\begin{remark}
Let us  naively think of the path space as countable. Writing $G_\Gamma$ for $G(\Gamma)$, $\tilde{\Gamma}:= \cR_\t(\Gamma)$ and setting $ C_\Gamma:= P_\Gamma  G_\Gamma$,
 \eqref{usignolo} is equivalent to $C_\Gamma=  C_{\tilde \Gamma}$, while
  \eqref{falco} reads
\begin{equation}
Q^y_{\Gamma} = P_{\Gamma}+ \frac{ y -y_F}{y_F}(P_{\Gamma}- C_{\Gamma})\,.
\end{equation}
The identity $E(G)=1$ would read $\sum_\Gamma C_\Gamma=1$.
The choice $C_\Gamma= \frac{1}{\mathcal{N}} \frac{P_\Gamma P_{\tilde{\Gamma}}}{P_\Gamma+P_{\tilde{\Gamma}}}$  as in \cite{PvdB} ($\cN$  being the normalization constant) would correspond to
\[ G_\Gamma=\frac{C_\Gamma}{P_\Gamma}=
 \frac{1}{\mathcal{N}} \frac{P_{\tilde{\Gamma}}}{
 P_\Gamma+P_{\tilde{\Gamma}}
 } \,,
\]
which is equivalent to
\[
\frac{1}{G_\Gamma} =\text{const} \left( 1 +  \frac{ d P}{dP\circ \mathcal{R}}\right)= \text{const} \left( 1 +e^{ Z}\right)\,.
\]
The above form of $G$  is exactly the choice \eqref{frontiera}.
\end{remark}

From now on $G$ is as in \eqref{frontiera}. For simplicity we  write
\begin{equation}
G =\frac{1}{\cN} \frac{1}{1+e^Z}\,,\qquad \cN =E\left[( 1+e^Z)^{-1}\right]\,.
\end{equation}   By \eqref{lady_bug} we have
\begin{equation}\label{adrian}
I_F(y) \leq H( Q^y|P)= E\Big[ \frac{dQ^y}{dP}\ln\frac{dQ^y}{dP}\Big] \,.
\end{equation}
Using that
$x\ln x=x-1+\frac{1}{2}\left(x-1\right)^{2}+o(\left(x-1\right)^{2})$, we obtain (recall that $E(G)=1$)
\begin{equation}\label{spignolo}
\begin{split}
I_F(y) &  \leq\frac{1}{2}\frac{(y-y_F)^{2}}{y_F^2}
E\bigl[ (1-G)^2 \bigr]+o\bigl( (y-y_F)^{2}\bigr)\\
&=\frac{1}{2}\frac{(y-y_F)^{2}}{y_F^2}
(E\bigl[ G^2\bigr]-1 \bigr)+o\bigl( (y-y_F)^{2}\bigr)\,.
\end{split}
\end{equation}
Now observe that
\begin{equation}\label{pavone1}
E[G^2]= \int P\circ \cR_\t (d\G) (G\circ \cR_\t)^2= E\Big[ \frac{d P \circ \cR_\t}{dP} (G\circ \cR_\t)^2 \Big]\,.
\end{equation}
Using \eqref{usignolo}  we get  that
\begin{equation}\label{pavone2}
 E\Big[ \frac{d P \circ \cR_\t}{dP} (G\circ \cR_\t)^2 \Big]=E\Big[ \frac{d P \circ \cR_\t}{dP}  G^2\cdot \Big (\frac{dP}{dP\circ \cR_\t} \Big)^2 \Big]=E\Big[  G^2 \frac{dP}{dP\circ \cR_\t}  \Big]=E[ G^2 e^Z]\,.
\end{equation}
As a byproduct of \eqref{pavone1} and \eqref{pavone2} we conclude that  $E[G^2]= E[ G^2 e^Z]$ and therefore
\begin{equation}\label{pavone3}
E[G^2]=\frac{1}{2} E[G^2(1+e^Z)]= \frac{1}{2\cN^2} E[(1+e^Z)^{-2}(1+e^Z)]=\frac{1}{2\cN}=\frac{1}{2 E\left[(1+ e^Z)^{-1}\right]}\,.
\end{equation}
Inserting the above identity in \eqref{spignolo} we get
\begin{equation}\label{spignolo100}
\begin{split}
I_F(y) \leq \frac{1}{2}\frac{(y-y_F)^{2}}{y_F^2}
\left(\frac{1}{2 E\left[(1+ e^Z)^{-1}\right]}-1 \right)+o\bigl( (y-y_F)^{2}\bigr)\,.
\end{split}
\end{equation}
We now claim that
\begin{equation}\label{mamma}
\frac{1}{\cN} =\frac{1}{E\left[(1+ e^Z)^{-1}\right]}\leq  1+ e^{E[Z]} = 1+ e^{\tau \s_{\rm naive}}
\end{equation}
(note that the identities in \eqref{mamma} follow from the definitions).
\ppp{By plugging \eqref{mamma} into \eqref{spignolo100} we get that
\begin{equation}\label{spignolocotto}
\begin{split}
I_F(y) \leq \frac{1}{4}\frac{(y-y_F)^{2}}{y_F^2}
\left( e^{\t \s_{\rm naive} }-1  \right)+o\bigl( (y-y_F)^{2}\bigr)\,,
\end{split}
\end{equation}
which implies \eqref{prosit}.
}

\ppp{Inequality \eqref{mamma}} corresponds to \cite[Eq. (17)]{PvdB} and  follows from a very tricky algebra in \cite[App. A]{PvdB} that we adapt to our  terminology.
Since $E[ e^{-Z}]=1$ (by the definition of $Z$), $P'$ defined as $dP'=\frac{1+e^{-Z}}{2} dP$ is a probability measure on $\Theta_\t$. By applying Jensen's inequality w.r.t. this probability $P'$ we have
\[ \ln \cN= \ln E \bigl[(1+e^Z)^{-1}\bigr] =\ln E\Big[\frac{1+e^{-Z}}{2} \frac{2 e^{-Z}}{(1+e^{-Z})^2}
\Big]\geq E\Big[\frac{1+e^{-Z}}{2} \ln \frac{2 e^{-Z}}{(1+e^{-Z})^2}
\Big]
\]
Since $E[e^{-Z} Z]=- E[Z]$ (by the definition of $Z$),  we have
\[
 E[Z] = -E\Big [\frac{ e^{-Z}-1}{2} Z\Big]
\]
Hence, setting $u:=e^{-Z}$, one gets a bound corresponding to \cite[Eq. (A.1)]{PvdB}:
\begin{equation}\label{ho_sonno}
\ln \cN + E[Z]\geq E\Big[ \frac{1+u}{2}\ln \frac{2 u }{(1+u)^2}+ \frac{u-1}{2} \ln u\Big]\,.
\end{equation}
Since $\frac{1+a}{2}\ln \frac{2 a }{(1+a)^2}+ \frac{a-1}{2} \ln a \geq (1-\ln 2)\frac{1+a}{2}-\frac{2 a}{a+1}$ for $a>0$ and since $dP'=\frac{1+u}{2} dP$ is a probability, we can lower bound the r.h.s. of \eqref{ho_sonno} by
\begin{equation}\label{ho_molto_sonno}
E\Big[  (1-\ln 2)\frac{1+u}{2}-\frac{2 u}{u+1}  \Big]= (1-\ln 2) -2E\Big[
 \frac{e^{-Z}}{e^{-Z}+1}\Big]= (1-\ln 2) -2\cN
\end{equation}
Since $1-\ln 2 - 2 a \geq \ln (1-a)$ for all $a\geq 0$, one concludes from \eqref{ho_sonno} and \eqref{ho_molto_sonno} that  $\ln \cN + E[Z]\geq \ln (1-\cN)$. This last estimate trivially implies \eqref{mamma}.

\subsection{Large deviations of $Q^{(n)}$}\label{gabriel}
Given $Q\in \mathcal{P}(\Theta_\t)$ we define, for $k,l\in V$,
\[
\begin{cases}
q_k = Q( \G_0=k)\,\\
q_{k,l}=Q(\G_0=k, \G_\tau=\ell)\,.
\end{cases}\]
 We let $\overline{q}=\bigl( q_{k,l}\bigr)_{(k,l)\in V\times V}$. When we want to  stress the dependence on $Q$, we write $q_k[Q]$, $q_{k,l}[Q]$, $\overline{q}[Q]$.  Recall that  $P$ is the law on $\Theta_\t$ of the random trajectory $\bigl( X_s\bigr)_{0\leq s\leq \t} $ when  $X_0$ has initial distribution $\pi_0$. We then set $p_{k,l}:= q_{k,l}[P]$ and $p_k:= q_k [P]$.

We consider the pair empirical measure
\begin{equation}
\overline{q}^{(n)}:=
\frac{1}{n} \sum_{j =0}^{n-1} \delta_{ (X_{j\tau }, X_{(j+1) \tau} )}\,,
\end{equation}
and observe that $\overline{q}^{(n)} :=\overline{q}[Q^{(n)}]$. By  \cite[Thm. IV.3]{dH}, $\overline{q}^{(n)}$ satisfies a LD principle with speed $n$ and rate functional $I_2$ defined as follows.   Let
 $\tilde{\mathfrak{M}}_1(V\times V)$ be given by the families
 \[ \overline c=(c_{k l})_{(k,l)\in V\times V}\]
  with
  \[ c_{kl}\geq 0\,, \qquad \sum _k \sum _l c_{k,l}=1\,, \qquad \sum _k c_{kl}= \sum _k c_{lk}\,.
  \]
If  $\overline q\in \tilde{\mathfrak{M}}_1(V\times V)$,  then
\begin{equation}\label{cammino} I_2\left(\overline q\right) :=
\sum_{k,l} q_{k,l}
\ln \frac{ q_{kl}}{  q_k P(X_\tau=l\,|\, X_0=k)}=
 \sum_{k,l} q_{k,l}
\ln \frac{ q_{kl}}{ p_{kl}}-\sum _{q_k}q_k \ln \frac{q_k}{p_k}
\,,
\end{equation}
otherwise  $I_2\left(\overline q\right) := +\infty$.

\begin{proposition}
$Q^{(n)}$ satisfies a LDP with  speed $n$ and  rate function
\begin{equation}\label{forma1}
 I(Q)= I_2\left(\overline q\right) + \sum_{k,l} q_{k,l} H[ Q_{kl}|P_{kl}]  \qquad  \forall Q \in \cP( \Theta_\t)\,,
 \end{equation}
 where
 \begin{itemize}
 \item $\overline q = \overline q[Q]$, $q_{k,l}= q_{k,l}[Q]$;
 \item
  $I_2$ is the pair empirical measure LD functional for the discrete time homogeneous Markov chain $(X_{n \tau})_{n \geq 0}$, which has invariant distribution $\pi_0$;
  \item
 $Q_{k,l}:= Q(\cdot| X_0=k\,,X_{\tau}=l)$;
 \item
 $P_{k,l}:= P(\cdot| X_0=k\,,X_{\tau}=l)$;
 \item $H[ Q_{kl}|P_{kl}] $ is the relative entropy of the probability  $Q_{kl}$ w.r.t. the probability $P_{kl}$.
 \end{itemize}
\end{proposition}
\begin{proof}
We only sketch the main idea which can be easily formalized. We will make  some abuse of notation  for the sake of intuition.  Recall that   $\overline q = \overline q[Q]$. Given $Q \in \cP(\Theta_\t)$ we have
\begin{equation}\label{salve1}
\begin{split}
 P\bigl ( \, Q^{(n)} = Q\, \bigr) = P\bigl(\, Q^{(n)} = Q, \; \overline q^{(n)} = \overline q \,\bigr)
 =
 P\bigl( \, Q^{(n)} = Q, \,|\,  \overline q^{(n)} =\overline q\bigr)P\bigl( \overline q^{(n)} = \overline q \bigr)
\end{split}
\end{equation}
By  \cite[Thm. IV.3]{dH} we have
\begin{equation}\label{salve2}
P\bigl( \overline q^{(n)} = \overline q \bigr)= e^{-n I_2(\overline q)}\,.
\end{equation} Consider the time interval $[0,n \t]$ as the union $\cup _{j=0}^{n-1} A_j$, where $A_j= [ j\tau, (j+1) \tau]$.
If we know that $\overline q^{(n)} = \overline q$, then for each pair $(k,l)$ we know that there are $q_{kl} n$ intervals $A_j$'s where the trajectory starts at $k$ and ends at $l$ (we call such a random set of intervals $\cA_{kl}$). If  we further condition on these intervals, then the random trajectories on $A_j$, with $A_j \in \cA_{kl}$, behave as $n q_{kl}$ i.i.d. random variables with value in $\Theta_\t$ and with distribution $P_{k,l}$.
 Moreover,  the random objects involved are independent when varying $(k,l)$. By applying Cram\'er's Theorem and the independence we conclude that
\begin{equation}\label{salve3}
 P\bigl( \, Q^{(n)} = Q, \,|\,  \overline q^{(n)} =\overline q\bigr)=\prod_{(k,l)} e^{- q_{kl} n H[Q_{k,l}|P_{k,l}]}\,.
\end{equation}
The thesis then follows as a byproduct of \eqref{salve1}, \eqref{salve2} and \eqref{salve3}.
\end{proof}
\bigskip
Note that, since in \eqref{forma1}, $\overline q=\overline q[Q]$ and $Q\in \mathcal{P}(\Theta_\t)$,  we get
that $\overline q \in \tilde{\mathfrak{M}}_1(V\times V)$ if and only if $\sum _k q_{kl} = \sum _k q_{lk} $ for each $l\in V$, which is equivalent to \eqref{condizione0}.
As a consequence, if $Q\in \mathcal{P}(\Theta_\t)$  fulfills \eqref{condizione0}
then
\begin{equation}\label{harry}
 I_2\left(\overline q\right) = \sum_{k,l} q_{k,l}
\ln \frac{ q_{kl}}{ p_{kl}}-\sum _{q_k}q_k \ln \frac{q_k}{p_k}\text{ where } \overline q=\overline q[Q]\,.
\end{equation}
By combining \eqref{forma1} and \eqref{harry} one easily gets  that
the LD rate functional $I(Q)$, for  $Q\in \mathcal{P}(\Theta_\t) $ fulfilling \eqref{condizione0},
can  be written as
\begin{equation}\label{forma20}
I(Q)=H(Q|P)-\sum _k q_k \ln\frac{q_k}{p_k}\,.
\end{equation}
We derive \eqref{forma20} for completeness. Given $k,l\in V$ we set $\Theta_\t(k,l):= \{ \G \in \Theta_\t\,:\,\G_0=k,\;\G_\t=l\}$.   Then, when $Q_{k,l}\ll P_{k,l}$ (the case  $Q_{k,l}\not \ll P_{k,l}$ can be treated easily)
\begin{equation}
\begin{split} H[Q_{k,l}|P_{k,l}]& =\int _{\Theta_\t(k,l)}  Q_{k,l} (d\G)\ln  \frac{ d Q_{k,l}}{dP_{k,l} } (\G)\\
&=\frac{1}{q_{k,l}} \int _{\Theta_\t(k,l)}  Q (d\G)\ln  \frac{ d Q }{dP } (\G)  - \ln \frac{q_{k,l}}{p_{k,l}}\,.
\end{split}
\end{equation}
By combining the above equation with \eqref{cammino} and \eqref{forma1} we get the \eqref{forma20}.

As a consequence we have
\begin{equation}
I(Q) \leq H(Q|P)
\end{equation}
 for any $Q\in \mathcal{P}(\Theta_\t) $  fulfilling \eqref{condizione0}.

\bigskip

\noindent
{\bf Acknowledgements}. A.F. and D.G. thank the Laboratoire J.A  Dieudonn\'e in Nice   for the kind hospitality and the University of Nice for the financial support.


\end{document}